\documentclass[12pt]{article}
\usepackage{times}
\usepackage{fullpage}
\usepackage{amsfonts,amssymb,amsmath,amsthm}
\usepackage{latexsym} 
\usepackage{gauss}
\usepackage{tikz}
\usepackage{tikz-cd}
\usepackage{calc}
\usepackage{url}
\usepackage{thmtools}
\usepackage{thm-restate}
\usepackage{hyperref}
\usepackage{makecell}
\usepackage[capitalise, nameinlink]{cleveref}
\usepackage{algorithm}
\usepackage{algorithmicx}
\usepackage{algpseudocode}
\usepackage{tabularx}
\usepackage{graphicx}

\newcommand{\qr}{\mathrm{QR}}
\newcommand{\qw}{\mathrm{QW}}
\newcommand{\su}{\mathrm{sum}}
\newcommand{\id}{\mathrm{Id}}
\newcommand{\N}{\mathbb{N}}

\newcommand{\R}{\mathbb{R}}
\newcommand{\Z}{\mathbb{Z}}

\newcommand{\floor}[1]{\lfloor #1 \rfloor}
\newcommand{\ceil}[1]{\left\lceil #1 \right\rceil}

\def\01{\{0,1\}}

\DeclareMathOperator*{\argmax}{arg\,max}
\DeclareMathOperator*{\argmin}{arg\,min}

\newcommand{\concat}{\ensuremath{+\!\!\!\!+\,}}
\newcommand{\Iset}{\mathcal{I}}
\newcommand{\indices}{\mathcal{I}}
\newcommand{\Jset}{\mathcal{J}}
\newcommand{\Kset}{\mathcal{K}}
\newcommand{\Lset}{\mathcal{L}}
\newcommand{\bldsf}{\textsc{Bipartite Longest Distinct Substring}}
\newcommand{\blds}{\textsc{BLDS}}

\newcommand{\prev}{\mathrm{prev}}
\newcommand{\suc}{\mathrm{succ}}
\newcommand{\rstart}{\texttt{rstart}}
\newcommand{\rend}{\texttt{rend}}
\newcommand{\rmiddle}{\texttt{rmiddle}}
\newcommand{\lstart}{\texttt{lstart}}
\newcommand{\lend}{\texttt{lend}}
\newcommand{\currentmax}{\texttt{currentMax}}

\newcommand{\kth}{k^{\scriptsize \mbox{{\rm th}}}}

\newcommand{\suppress}[1]{}

\newcommand{\ket}[1]{|#1\rangle}

\newcommand{\Int}[1]{\mathrm{Int}(#1)}
\newcommand{\cont}[1]{\mathbf{c-}#1}

\newcommand{\upto}{\mathbin{:}}

\DeclareMathOperator{\polylog}{polylog}

\DeclareMathOperator{\Maj}{\textsc{Maj}}

\newcommand{\MP}{\textsc{Matrix Product}}

\newcommand{\Met}{\textsc{Metricity}}
\newcommand{\MaxT}{\textsc{Maximum Triangle}}

\newcommand{\MaxSubMf}{\textsc{Maximum Submatrix}}
\newcommand{\MaxSubM}{\textsc{MSM}}

\newcommand{\MaxFourCf}{\textsc{Maximum 4-Combination}}
\newcommand{\MaxFourC}{\textsc{M4C}}

\newcommand{\SSSTf}{\textsc{Single Stock Single Transaction}}
\newcommand{\SSST}{\textsc{SSST}}
\newcommand{\APSPf}{\textsc{All-pairs Shortest Paths}}
\newcommand{\APSPc}{\bf{APSP}}
\newcommand{\APSP}{\textsc{APSP}}
\newcommand{\LSICf}{\textsc{Longest Substring of Identical Characters}}
\newcommand{\LSIC}{\textsc{LSIC}}
\newcommand{\LISstf}{\textsc{Longest Increasing Substring}}
\newcommand{\LISst}{\textsc{LISst}}
\newcommand{\ISf}{\textsc{Increasing Subsequence}}
\newcommand{\LISf}{\textsc{Longest Increasing Subsequence}}
\newcommand{\IS}{\textsc{IS}}
\newcommand{\ldsf}{\textsc{Longest Distinct Substring}}
\newcommand{\lds}{\textsc{LDS}}

\newcommand{\LMSf}{\textsc{Length Minimal Substring}}
\newcommand{\MSRf}{\textsc{Minimal String Rotation}}
\newcommand{\MSf}{\textsc{Minimal Suffix}}
\newcommand{\SSMTf}{\textsc{Single Stock Multiple Transactions}}
\newcommand{\Recognizing}{\textsc{Recognizing}}
\newcommand{\LZSf}{\textsc{Longest} $20^*2$ \textsc{Substring}}
\newcommand{\LZS}{\textsc{L}$20^*2$\textsc{S}}
\newcommand{\SSMT}{\textsc{SSMT}}
\newcommand{\SiSf}{\textsc{Signed Sum}}
\newcommand{\SiS}{\textsc{SS}}
\newcommand{\MSSTf}{\textsc{Multiple Stocks Single Transaction}}
\newcommand{\MSST}{\textsc{MSST}}
\newcommand{\KCf}{\textsc{Klee's Coverage}}
\newcommand{\KMf}{\textsc{Klee's Measure}}
\newcommand{\CSf}{\textsc{Common Subsequence}}
\newcommand{\RMP}{\textsc{Recursive Max Pooling}}

\newtheorem{theorem}{Theorem}

\newtheorem{lemma}[theorem]{Lemma}
\newtheorem{corollary}[theorem]{Corollary}
\newtheorem{proposition}[theorem]{Proposition}

\newtheorem{fact}[theorem]{Fact}

\newtheorem{openq}{Open Question}
\newtheorem*{rtheorem}{Theorem}

\theoremstyle{definition}

\newtheorem{problem}[theorem]{Problem}
\newtheorem{definition}[theorem]{Definition}

\begin{document}
\title{On the quantum time complexity of divide and conquer}
\author{
  ~~~~~~~~
  Jonathan Allcock\thanks{Tencent Quantum Laboratory, Hong Kong. \url{jonallcock@tencent.com}} ~~~ \and
  Jinge Bao\thanks{Centre for Quantum Technologies, National University of Singapore. \url{jbao@u.nus.edu}} ~~~ \and
  Aleksandrs Belovs\thanks{Faculty of Computing, University of Latvia. \url{stiboh@gmail.com}} ~~~~~~~~ \and
  Troy Lee\thanks{Centre for Quantum Software and Information, University of Technology Sydney. \url{troyjlee@gmail.com}} \and
  Miklos Santha\thanks{Centre for Quantum Technologies and MajuLab, National University of Singapore. \url{miklos.santha@gmail.com}}
}

\date{}
\maketitle

\begin{abstract}
We initiate a systematic study of the time complexity of quantum divide and conquer algorithms for classical problems.
We establish generic conditions under which search and minimization problems with classical divide and conquer algorithms are amenable to quantum speedup and apply these theorems to an array of problems involving strings, integers, and geometric objects. They include $\ldsf$, $\KCf$, several optimization problems on stock transactions, and $k$-\ISf.
For most of these results, our quantum time upper bound matches the quantum query lower bound for the problem, up to polylogarithmic factors.
\end{abstract}

\section{Introduction}
\label{sec:intro}
Divide and conquer is a basic algorithmic technique that gained prominence in the 1960s via a series of beautiful and powerful algorithms for fundamental problems including integer multiplication~\cite{KO62}, sorting~\cite{Hoa62}, the Fourier transform~\cite{CT65} and matrix multiplication~\cite{Str69}. In fact, the technique of divide and conquer was already used much earlier: Gauss designed the first fast Fourier transform algorithm in 1805, published only after his death~\cite{G1876, HJB84}, and von Neumann invented mergesort in 1945~\cite{GN47}.

Divide and conquer algorithms do not have a generic mathematical description unlike, for example, greedy algorithms.
Similarly, there are no known combinatorial structures on which they achieve optimality, unlike greedoids where the greedy solution is optimal for a large class of objective functions~\cite{KL84}. 
In order to capture the variety of divide and conquer algorithms, it is helpful to keep three examples in mind: quicksort, mergesort, and a divide and conquer algorithm for the \SSSTf\ (\SSST) problem, where the goal is to compute $\argmax_{i < j} A_j - A_i$ for a given array $A \in \mathbb{Z}^n$.  
We assume that the reader is familiar with quicksort and mergesort.  
The divide and conquer algorithm for \SSST\ is based on the principle that the indices achieving the maximum are either both 
contained in the left half of the array, both contained in the right half of the array, cases that can be solved recursively, or one is contained in the left half and one is the right half. The latter case can be directly solved because in this case $A_j - A_i$ is maximized when $A_j$ is the maximum value in the second half of $A$ and $A_i$ is the minimum value in the first half of $A$. 

With these examples in mind, we break down the work in a divide and conquer algorithm into four C's:
create, conquer, complete, and combine. 
In the {\em create} (or divide) step the algorithm constructs the subproblems to be solved. 
In mergesort, the construction of subproblems is trivial as this simply involves dividing a string into left and right halves.  In quicksort, 
on the other hand, the create step is where the main work of partitioning the input into elements at most the pivot and at least the pivot takes place.
In the {\em conquer} step the algorithm (typically recursively) solves the subproblems created in the previous step. In the {\em complete} step the 
algorithm computes anything that is needed to solve the original problem but not contained in the solution to the subproblems.  \SSST\ is a prime example of this step, as here in the complete step one solves the special case of finding the maximum profit when buying in the first half of the array and selling in the second half. 
Finally in the {\em combine} step the algorithm integrates the solutions to the subproblems and from the complete step to solve the original problem.
A canonical example of this is the merge step of mergesort. \cref{table:dandc} summarizes this discussion.
\begin{table}
\begin{tabularx}{\linewidth}{|l|X|X|}
  \hline
  & Description & Example  \\
  \hline
Create & Create subproblems. & Partition step of quicksort. \\
\hline
Conquer & Solve subproblems. & Recursive call to quicksort. \\
\hline
Complete & Compute additional information not covered by subproblems. & 
In \SSST, the maximum profit from buying in the first half, and selling in the second half.\\
\hline
Combine & Combine subproblem solutions to solve the original problem. & 
The merge step of mergesort. \\
\hline
\end{tabularx}
\caption{Breakdown of the steps in a divide and conquer algorithm.}
\label{table:dandc}
\end{table}

Classical divide and conquer algorithms are usually analyzed in the RAM model where accessing a memory register containing an
element of the input string takes constant time. Our quantum algorithms will be given in 
two quantum memory models. In both cases, we assume coherent access to the memory register, that is, the index register can store a superposition of addresses. The first model we consider is $\mathsf{QRAM}$, where the access to data is read-only, and we further assume that the data stored is classical, i.e., the memory register is not in superposition. The second model we consider is $\mathsf{QRAG}$, where both reading and writing to memory are permitted, and we assume the memory register itself can be in superposition (indeed we require this in some of our algorithms). 
By a {\em query} we mean any of these memory operations. 
While we also will derive query complexity results about our algorithms, our main emphasis will be on time complexity.
For a memory of size $N$, we denote the time required to perform a $\mathsf{QRAM}$ or $\mathsf{QRAG}$ query by parameters $\qr_N$ (quantum reading) and $\qw_N$ (quantum writing), respectively.
Standard one and two-qubit quantum operations are counted as one time step. 
The exact definition of these models 
is given in \cref{subsec:memory}.

\subsection{Our contributions}
\begin{table}
\centering
\resizebox{\textwidth}{!}{
\begin{tabular}{| l | c | c |}
\hline
& Classical Time & Quantum Time \\
\hline
&&\\[-1.0em]
\ldsf  & $\widetilde O(n)$   & $\widetilde O(n^{2/3})$ \\ 
\hline
&&\\[-1.0em]
\KCf\  in $\mathbb{R}^d$ & $O(n^{d/2}), d\ge 3$~\cite{Chan13} & $O(n^{d/4+\varepsilon}), d \ge 8$ \\
\hline
&&\\[-1.0em]
\SSSTf\ & $O(n)$  & $O(\sqrt{n} \log^{5/2}(n))$ \\
\hline
&&\\[-1.0em]
\makecell{$k$-\ISf \\ $k$-\SiSf} & \makecell{$ O(n \log \log k)$ \cite{CP08} \\ $\widetilde{O}(n)$} & $O(\sqrt{n} \log^{k+1}(n) (\log \log n)^{k-1})$ \\
\hline
&&\\[-1.0em]
\MaxFourCf & $O(n^3)$ \cite{BDT16} & $O(n^{3/2} \log^{5/2} (n))$ \\
\hline
\end{tabular}
}
\caption{Applications of our divide and conquer technique in the third column, assuming quantum memory access times $\qr_n, \qw_n=O(\log^2 n)$ (see \cref{sec:instances}).
For each sublinear quantum upper bound, there is a simple quantum query lower bound
that matches up to a polylogarithmic factor. }
\label{table:applications}
\end{table}

In this work, we focus on divide and conquer algorithms where the combine step is either the OR function or minimization/maximization, meaning that it is amenable to a quantum speedup by Grover search or quantum minimum/maximum finding.  
When the complete step is also relatively simple, we can show a generic theorem that transforms a classical divide and conquer algorithm into a quantum one and bounds its time complexity.  We show two versions of this theorem, \cref{thm:recgen} and \cref{thm:rec} depending on the nature of the create step.

Our generic quantum divide and conquer theorems have several nice features.  
The first is that, to apply them, it suffices to come up with a \emph{classical} divide and conquer algorithm satisfying the conditions of the theorem.  The quantization of this algorithm is then completely handled by the theorem. This can make it easier to find applications of 
these theorems.  The second is that these theorems give bounds on \emph{time complexity}, not just the query complexity. 
To the best of our knowledge, the quantum time complexity of the divide and conquer method has not been systematically studied before.  Working with time complexity also lets us apply these theorems to a greater range of problems, e.g.\ those where the best-known quantum algorithm requires super-linear time.

As divide and conquer algorithms are typically recursive, our theorems must handle the error probability corresponding to the composition of a super-constant number of bounded-error algorithms. 
While the composition of bounded-error quantum algorithms is now well understood in the quantum query complexity setting, much less work has focused on this in the time complexity setting. 
We show several basic results about the time complexity of the composition of quantum search or minimum/maximum finding over bounded-error subroutines. 
In most cases, these results are relatively easy adaptations of analogous algorithms in the query setting.  
However, we feel these are fairly fundamental algorithmic primitives that will find further application in the future.

The statement of our quantum divide and conquer theorems is rather technical and we delay further details to \cref{sec:bottomup} and \cref{sec:disj}. 
In the rest of this section, we describe applications of our generic theorems to different problems.
A summary of the major applications can be found in \cref{table:applications}, where we assume that the cost of the two memory access gates is $O(\log^2 n)$. 

Some problems have a simple create step, by which we mean that the locations of the subproblems are identical for all inputs of the same size. 
Problems with a simple create step are dealt with in \cref{sec:bottomup}, \cref{sec:is} and \cref{sec:apsp}.
The quantum algorithm we give for \SSSTf \ turns out to be quite paradigmatic, and we are able to use a very similar approach for 
the following three problems. 
In the \LISstf \ problem (\LISst), we have to find the longest increasing substring in an array of integers.
In the \LSICf \ problem (\LSIC), we are looking for the longest substring of identical characters in 
a string over some finite alphabet. 
Finally, 
in the \LZSf \ problem (\LZS) we have to find the longest substring belonging to $20^*2$ in a string
from $\{0,1,2\}^*$.
All these problems 
can be solved by classical divide and conquer along the lines of the algorithm we sketched for \SSSTf, 
in time  $O(n \log n)$, and they require time $\Omega(n)$.
Our algorithms achieve an almost quadratic quantum speed up.  For the rest of the paper, for $k \geq 1$, we define the function 
\[
\lambda_k (n, m) = \min \{\log^k n + m , \log^{(k+1)/2} (n) (\log \log n)^{k-1} + \log^{(k-1)/2} (n) (\log \log n)^{k-1} \cdot m \}.
\]
Observe that for $k=1$ we have $\lambda_1 (n, \qr_n) =  \log n + \qr_n,$ and for 
$k=2$ we have $\lambda_2 (n, \qr_n) = \min \{\log^2 n+ \qr_n , \log^{3/2} (n) \log \log n + \sqrt{\log n} \log \log n \cdot \qr_n \}.$

\begin{rtheorem}[\cref{thm:three} {\normalfont restated}]
The quantum query and time complexities of the problems \SSST, \LISst, \LSIC \ and \LZS \ 
are respectively $O(\sqrt{n \log n})$ and 
$O({\sqrt{n \log n} \cdot \lambda_2 (n, \qr_n}))$.
\end{rtheorem}
 
One generalization of the stock problem is $d$-\MSSTf, for $d \geq 1$,  where we want to 
find $\max_{i < j} A_j - A_i$ in a
$d$-dimensional array $A$. Here, by definition, $i<j$ when $i_k < j_k$, for $1 \leq k \leq d.$
Our quantum algorithm easily generalizes to that problem.

\begin{rtheorem}[\cref{thm:ssmt} {\normalfont restated}]
The quantum query complexity of $d$-\MSSTf \ is
$O( (n \log n)^{d/2})$ and
its time complexity is
\[
O \left ((n \log n)^{d/2} \cdot \min \{\log^{d+1} n + \qr_{n^d} , \log^{1 + d/2} (n) \log \log n + \log^{d/2} (n) \log \log n \cdot \qr_{n^d} \} \right).
\]
\end{rtheorem}

In the $k$-\ISf \ problem, we would like to decide if, in an array of integers, there is a subsequence of $k$ increasing numbers.
As discussed in~\cite{CKKSW22}, this is the natural parametrization of the classically well-studied \LISf \ problem, closely related to patience sorting. There are $O(n \log n)$ query classical dynamic programming algorithms solving \LISf, and it is easy to show an $\Omega(n)$ quantum query lower bound for it. This implies that no substantial quantum improvement can be obtained for $k$-\ISf \ when $k$ is unbounded, and makes the case of constant $k$ an interesting research question. In~\cite{CKKSW22} an $O(\sqrt{n \log^{3(k-1)} n})$ quantum query algorithm was obtained
for $k$-\ISf, and we improve this result by a factor of $O(\log^{k-1} n)$ and implement it time-efficiently.
It turns out that a very similar quantum algorithm can solve the $k$-\SiSf \ problem with the same complexity.
In this problem, given an array of integers and a sign pattern $\varepsilon \in \{-1,1\}^k$, 
we want to maximize the signed sum $\sum_{m =1}^k  \varepsilon_m A_{i_m}$, for  $k$ indices satisfying
$i_1  < i_2< \ldots < i_k$.

\begin{rtheorem}[\cref{thm:k-is} and \cref{thm:k-sis} {\normalfont restated}]
There are quantum algorithms that solve $k$-\ISf \ and $k$-\SiSf \ with $O( \sqrt{n \log^{k-1} n})$ queries. 
The  time complexity of the algorithms is $O(\sqrt{n \log^{k-1} n} \cdot \lambda_k(n,\qr_n))$.
\end{rtheorem}

\cref{sec:disj},\cref{sec:lds}, and \cref{sec:klee} deal with problems with simple complete steps, and hence the main work lies in the create step. One interesting example of this is the \ldsf\ ($\lds$) problem.  Given a string $a \in \Sigma^n$ over an alphabet $\Sigma$, this 
problem is to find the longest contiguous substring $a_i a_{i+1}\cdots a_j$ where all letters are unique.  The famous Element Distinctness problem is a special case of $\ldsf$ where the task is to determine if the length of a longest distinct substring is equal to the length of $a$ itself.  
Ambainis \cite{ambainis2007quantum} famously gave a quantum walk algorithm showing the time complexity of element distinctness is $\tilde O(n^{2/3})$, which is tight \cite{aaronson2004quantum}. 
We apply our divide and conquer theorem, (\cref{thm:rec}), in conjunction with a novel classical divide and conquer algorithm for \lds, to show that a quantum algorithm can solve \lds \ in time $\widetilde{O}(n^{2/3}\cdot \qw_{O(n)})$. 
Thus, up to logarithmic factors, finding the longest distinct substring has the same quantum time complexity as element distinctness. 

\begin{rtheorem}[\cref{thm:lds} {\normalfont restated}]
The quantum time complexity of the \ldsf\ problem is $\widetilde{O}(n^{2/3}\cdot \qw_{O(n)})$.
\end{rtheorem}

Another example with a simple complete step is the $\KMf$ problem from computational geometry, which asks to compute the volume of the union of 
axis-parallel hyperrectangles in $d$-dimensional real space. In the special case of the \KCf \ problem, the question is to decide if the union of the hyperrectangles covers a given base hyperrectangle.
In 2-dimensions, the classical complexity of the $\KMf$ problem is $O(n \log n)$~\cite{Klee77}, and for any constant $d \geq 3$, Chan~\cite{Chan13} has designed an $O(n^{d/2})$ time classical algorithm for it.
We give a quantum algorithm for $\KCf$ that achieves almost quadratic speedup over the classical divide and conquer algorithm of Chan~\cite{Chan13}, when $d \geq 8$.

\begin{rtheorem}[\cref{thm:klee} {\normalfont restated}]
For every constant $\varepsilon > 0$, the quantum time complexity of the \KCf \ problem is $O(n^{d/4 + \varepsilon} \cdot \qw_{O(n^{d/2+\epsilon})})$ when $d \geq 8$, and is ${O}(n^2 \log n \cdot \qw_{O(n^{d/2+\epsilon})})$ for $5\le d\le 7$.
\end{rtheorem}

Perhaps not surprisingly, our results have some consequences for fine-grained complexity, in particular for the class $\APSPc$ of problems that are solvable in time $\widetilde{O}(n^{3})$ on a classical computer and are sub-$n^3$ equivalent to the \APSPf \ problem in the sense that either all of them or none of them admit an $O(n^{3- \varepsilon})$ algorithm, for some constant $\varepsilon > 0$.
$\APSPc$ is one of the richest classes in fine-grained complexity theory~\cite{VW18, Vas19} and, in particular, contains various path, matrix, and triangle problems. 

Quantum fine-grained complexity is a relatively new research area~\cite{ACLWZ20, BPS21, BLPS22} where one possible direction is to study the quantum complexity of problems in the same classical fine-grained equivalence class.
Indeed, the work~\cite{ABLPS22} specifically considered $\APSPc$. 
Of course, there is no guarantee that classically equivalent problems remain equivalent in the quantum model of computing, and this is indeed the case for $\APSPc$. All problems in the class receive some quantum speedup, but the degree of speedup can differ from problem to problem. It turns out that many of the problems in $\APSPc$ can be solved either in time $\widetilde{O}(n^{5/2})$ or in time $\widetilde{O}(n^{3/2})$ by simple quantum algorithms and, concretely, $\APSPf$ falls in the former category. 
We consider the quantum complexity of two problems from the class $\APSPc$:
$\MaxFourCf$ and $\MaxSubMf$, both of which take an $n \times n$ matrix $B$ as input.
We want to compute, for the former, the maximum of $ B_{ik} + B_{j \ell} - B_{i \ell} - B_{jk} $ and, for the latter, the maximum of $\sum_{i \leq u \leq j, k \leq v \leq \ell} B_{uv}$, under the conditions $1 \leq i  \leq j \leq n$ and $1 \leq k  \leq \ell \leq n$. Our quantum algorithm for $\MaxFourCf$ uses the quantum divide and conquer method designed for \SSSTf.

\begin{rtheorem}[\cref{thm:msm} and \cref{thm:mfc} {\normalfont restated}]
The quantum time complexity of \MaxSubMf\ is $O(n^{2} \log n)$ and its query complexity is $\Omega(n^2).$
The quantum time complexity of \ \MaxFourCf ~is ${O}(n^{3/2} \log^{5/2} (n)).$
\end{rtheorem}

\subsection{Our techniques} 
As mentioned, our techniques all suppose that the combine step of the divide and conquer is search or minimization and, moreover, the complete or the create steps (or both) are relatively simple.
We now describe in more detail how we are able to exploit these properties.

{\bf Search or minimization combine step.} The key technical tool we will use repeatedly is based on  a relatively old result of H{\o}yer, Mosca and de Wolf~\cite{hoyer2003quantum}, see \cref{fact:qseo}. 
Their result is stated for query complexity and essentially says that if, for $J$ boolean functions $f_1, \ldots , f_J:\{0,1\}^n\rightarrow\{0,1\}$ there exists a quantum algorithm $F$ that on input $\ket{j} \ket{x} $ correctly computes $f_j(x)$ with probability at least $8/10$, then 
there exists a quantum algorithm which uses ${O}(\sqrt{J} )$ repetitions of $F$ and with probability at least $9/10$ finds a marked index, that is $1 \leq j \leq J$ such that $f_j(x)=1$, if there is one.
The main point here is that the number of repetitions is of the same order of magnitude as one would need when $F$ does the computation without error.
We analyze the time complexity of the above algorithm and derive in \cref{cor:seo} that it is ${O}({\sqrt{J} ( \log J + \tau}))$,  where $\tau$ is the time complexity of $F$.
A similar result was also known for the query complexity of finding the index of a minimum element (see \cref{fact:qmeo}), and we obtain the analogous result for the time complexity in \cref{cor:meo}. In the remaining part of the section, while we only speak about minimization algorithms, everything is valid for search problems as well.

What can we say about the time complexity of $F$ if, for every $1 \leq j \leq J$, we have at our disposal an algorithm $A_j$ of known complexity computing $f_j(x)$? Let us call these {\em base} algorithms, and for the simplicity of discussion let us suppose that the time complexity of every base algorithm is $S + q \cdot \qr_n$, where $S$ is the total number of one and two-qubit gates and $J, S, q$ depend on $n$. If there is no particular relation between the $J$ base algorithms (that is, they can be very different), we do not have a particularly clever implementation of $F$. One possibility is to compose the quantum circuits computing them, where the non-query gates of $A_j$ are controlled by $j$. The query gates can be executed without control, giving an overall complexity of $F$ of $O(JS + q \cdot \qr_n)$, and yielding a minimization algorithm of complexity $O(\sqrt{J}(JS + q \cdot \qr_n))$.
Another trivial way to solve the minimization problem is to reduce the error of each base algorithm via $\log J$ repetitions and then classically compute the minimum. The complexities of these approaches are stated in \cref{lem:two}.

However, there is one situation where we can do much better, and this is exactly the case of recursive algorithms such as divide and conquer.  In this case, the base algorithms correspond to the recursive calls and are therefore the same by definition. Thus, for implementing $F$ there is no need to use controlled operations or to repeat the base algorithms, provided that we are able to determine, for every $j$, the memory indices where the  $\mathsf{QRAM}$ gates for $A_j$ are executed. 
There are two broad cases we consider, corresponding to simple create and simple complete steps:

{\bf Simple create step.} In some cases, such as in the \SSSTf\ problem, the create step is independent of the input. 
Suppose that the input length is a power of $2$.
If we unfold the successive recursive steps for this problem, then it is easy to see that, for every input $A$ of length $n$, there exists $t \in [\log n]$, such that if we partition $[n]$ into $2^t$ consecutive intervals of length $n/2^t$ then one of these intervals $C$ contains both $i$ and $j$, with $i<j$, such that $A_j - A_i$ is an optimal solution.
Moreover, $i$ is in the left half $C_{\ell}$, and $j$ is in the right half $C_r$ of $C$, and therefore under this assumption they can be found easily in time $\sqrt{n/2^t} (\log n + \qr_n)$.

Now the main point is that, for a given $t$, we don't know which interval $C$ contains the solution, but to find this good interval we can use quantum maximum finding with an erroneous oracle.
These $2^t$ intervals are consecutive and therefore the elements in all intervals can be indexed uniformly, once we settle the first $t$ bits specific to each interval.
For every $t$, this results in a uniform cost $\sqrt{n} (\log n + \qr_n)$.
After this, we still have to search over the $\log n$ possible values of $t$, again using \cref{cor:meo}.
This time there are additional costs for implementing $F$, captured in \cref{lem:two}, because the above procedures are quite different for different values of $t$. However, as the search is only over a domain of logarithmic size the additional cost is not substantial.

We call the resulting divide and conquer method {\em bottom-up} since the method actually completely eliminates the recursive calls, and we give a general statement about this approach in \cref{thm:bu_easy}. Avoiding recursion is possible because we know the subproblems in advance, as they are independent of the input.
Moreover, under the condition that the solution is contained in some interval $C$, in our case the additional information that $i$ is in $C_{\ell}$ and $j$ is in $C_r$ makes the solution easier. 
This, of course, isn't always the case.
Consider the element distinctness problem on an interval of size $n$. The additional knowledge that the two colliding elements are in different halves of the interval does not make the task of finding them easier.  
But it does help for \SSSTf \ as well as for \LISstf, \LSICf, \LZSf \ and $d$-\MSSTf.
In fact, to some extent, we can even generalize the method to $k$-ary relations, for $k > 2$, which is illustrated by our algorithms for $k$-\ISf \ and $k$-\SiSf.

{\bf Simple complete step.}  
In this case, the main work of the algorithm is done in the create step. We give two relatively generic theorems describing situations where the index problem can be handled well.

The first situation concerns \emph{constructible instance} problems, where the recursive calls are executed on subproblems that are explicitly constructed by the create step. Such a create step may be expensive to perform and, in particular, all the inputs of the recursive calls (which we call {\em constitutive strings}) have to be written down, likely taking at least linear time. This implies that this approach can only yield nontrivial results for time, and not for query complexity.
We can suppose that the constitutive strings are written in a way that the indices of the memory used by the $J$ functions calls $A_j$ differ only in the first $\log J$ bits where, for every $j$, the algorithm $A_j$ has $j$ in binary. Copying $j$ into the appropriate place in the index of the $\mathsf{QRAG}$ gates is most likely a negligible cost. \cref{thm:recgen} gives the exact statement. An application of this result is the quantum algorithm for \KCf \, where the create step requires quadratic time.
However, in high dimensions, the time of the homogeneous recursion, which doesn't include the create step, is much larger and therefore we are able to achieve
an almost quadratic speed up over the classical algorithm.

The second situation concerns {\em $t$-decomposable instance} problems, where the recursive calls are made on subsequences of the input. In this case, the create step does not need to create the constitutive strings but must determine the indices that delimit the constitutive strings in the input. This can, in principle, take much less than linear time. 
\cref{thm:rec} gives the exact statement for this approach. As an application of this, by identifying a problem, that we call \bldsf, as a $t$-decomposable instance problem whose create and complete functions can be computed in time $\widetilde{O}(n^{2/3}) \cdot \qw_{O(n)}$, we are able to give a quantum algorithm for solving the \ldsf\ problem in time  $\widetilde{O}(n^{2/3}) \cdot \qw_{O(n)}$.  

\subsection{Previous work}
Recently, several quantum divide and conquer algorithms were presented for various string problems.
In the first of these papers, Akmal and Jin~\cite{AJ22} considered the $k$-\LMSf \ problem for $k \geq n/2$, where given a string $a$ of length $n$ over some finite
alphabet with a total order, the output is a substring $v$ of $a$ of length $k$ such that for every substring $w$ of $a$ of length $k$, we have $v \leq w$, for lexicographic ordering among strings. In the decision version of the problem, the input also includes a string $v$ of length $k$ and the question is whether $v$ is lexicographically smallest among the $k$-length substrings of $a$. 

The algorithm in~\cite{AJ22} works in time $\sqrt{n} 2^{O(\log^{2/3} n)}$, and the high-level structure of the proof goes along the lines of our \cref{thm:rec} for $t$-decomposable instance minimizing problems, but without involving our \cref{cor:meo} to deal efficiently with the errors of the recursive calls.
The combinatorial contents of the proof, however, are quite different.
In~\cite{AJ22} it is also shown that the problems \MSRf \ and \MSf \ are easily reducible to $k$-\LMSf \ and therefore have the same quantum time complexity upper bound. In the former problem, one is looking for $j \in [n]$ such that $a[j:n]a[1:j-1] \leq a[i:n]a[1:i-1]$, for all $i \in [n]$, while in the latter problem one looks for $j \in [n]$ such that $a[j:n] < a[i:n]$, for all $i \neq j$. The decision versions of these problems, when $j$ is also given in the input, are defined naturally.

In the query complexity model, two papers improved on these results. 
Childs et al. in ~\cite{CKKSW22} showed that the decision versions of the above problems can be solved in $O(\sqrt{n} \log^{5/2} n)$ queries.
For \MSRf \ Wang~\cite{Wan22} improved this to $\sqrt{n} 2^{O(\sqrt{\log n})}$, and for the decision version of the problem further improved this to $O(\sqrt{n} \log^{3/2} n \sqrt{\log \log n})$. Interestingly, Wang uses \cref{fact:qmeo} on bounded-error oracles in his argument, but the paper does not deal with time complexity.

Childs et al.~\cite{CKKSW22} recently proposed a powerful and general framework in the query complexity model for quantum divide and conquer algorithms computing boolean functions. Their framework includes, for example, the case where the complete step is computed by an AND-OR formula in the recursive calls and some auxiliary boolean functions. Besides the problems from the previous paragraph, it is also applied to give algorithms for $k$-\ISf \ and for $k$-\CSf \ where one has to decide whether two strings share a common subsequence of length $k$.
This framework makes use of specific properties of the query model (for example, the equality of the query complexity and the adversary bound), and it is unlikely that it can directly yield results for time complexity. With respect to this framework ours has the following advantages:
\begin{itemize}
\item
It deals with time complexity and not only with query complexity.
\item
As a consequence, it can also deal with problems of super-linear complexity.
\item
It can handle minimization problems and not just boolean functions.
\end{itemize}
Finally we mention that for the $k$-\ISf \ problem we improve their query complexity result of $O(\sqrt{n} \log^{3(k-1)/2} n)$ to $O(\sqrt{n} \log^{(k-1)/2} n)$.

After our paper was completed, we learned that Jeffery and Pass are preparing a manuscript~\cite{JP23} on time-efficient quantum divide and conquer using quantum subroutine composition techniques \cite{Jeff22}.

\section{Preliminaries}
\label{sec:prelim}
\subsection{Notation}
We denote by $\widetilde{O}({f(n))}$ the family of functions of the form $f(n) \polylog(n)$.
For a positive integer $n$, we denote by $[n]$ the set $\{1, \ldots , n \}$.
Let $\Sigma$ be a finite set and $a = a_1 \ldots a_{n} \in \Sigma^n$. 
We define the {\em length} of $a$ as $n$, and we denote it by $|a|$.
We call $a_{i_1} a_{i_2} \ldots a_{i_j}$ a {\em subsequence} of $a$, for any $1 \leq i_1 < i_2 < \ldots < i_j \leq n$. A {\em substring} is a subsequence of consecutive symbols, and for $1 \leq i \le j \leq n$, we use $a[i\upto j]$ to denote the substring $a_i a_{i+1} \ldots a_j$ of $a$.  If 
$i > j$ then $a[i \upto j]$ is the empty string. 
For two strings $a,b \in \Sigma^*$ we denote the concatenation of $a$ and $b$ by $a \concat b$.
For $n^* \in \N$ we define $\Sigma^{\leq n^*} = \bigcup_{n \leq n^*} \Sigma^n$.
When $\Sigma \subseteq \Z$, we use capital letters $A,B,C...$ for elements of $\Sigma^n$ or $\Sigma^{n \times n}$, that is, for arrays and matrices.
For ease of notation, we often neglect taking floors and ceilings in our computations, however, this never affects the correctness of the asymptotic results.

\subsection{Quantum computational models}
\label{subsec:memory}
As in a number of other papers on quantum algorithms~\cite{Amb07, BJLM13, ABIKPV19, AHJKS22}, we use the standard quantum circuit model of all single qubit gates and the two-qubit $\mathsf{CNOT}$ gate, augmented with random access to quantum memory. By quantum memory, we mean a specific register of the quantum circuit
where the input can be accessed in some specific way. 

\subsubsection{$\mathsf{QRAM}$ and $\mathsf{QRAG}$}

We will use two memory models. In the more restricted $\mathsf{QRAM}$ model the memory can only be accessed for reading.
Formally, for any positive integer $N$, we define the  $ \mathsf{QRAM}_N$ gate as
   \begin{align*}
        \mathsf{QRAM}_N |i, e, x \rangle = |i, e \oplus x_i,  x\rangle,
    \end{align*}
where $i\in[N], e, x_1, \ldots x_N \in\{0,1\}^r$.    
We refer to the three registers involved in a $ \mathsf{QRAM}_N$ gate as the memory {\em index}, the memory {\em output}, and the memory {\em content} registers.
The memory index and output registers can be in superposition, but the memory content register always only contains the input string.
The memory content register can only be accessed via the $\mathsf{QRAM}_N$ gates and these gates can only applied to the memory registers.
In the more general $\mathsf{QRAG}$  model the memory content register can also be accessed for writing.

To formalize this is, for any positive integer $N$, we define the $\mathsf{QRAG}_N$ gate by
    \begin{align*}
        \mathsf{QRAG}_N |i, e, x \rangle = |i, x_i, x_1,\dots,x_{i-1},e,x_{i+1}\dots, x_N \rangle, 
    \end{align*}
where
$i\in[N]$ and $e,x_1,\dots,x_{N} \in\{0,1\}^r.$
Similar to the other $\mathsf{QRAM}$ case, here the memory content register can only be accessed via $\mathsf{QRAG}_N$ gates and these gates can only applied to the memory registers.
Since writing to memory is permitted, the memory content register can also be in quantum superposition.

In our running time analyses, single and two-qubit gates count as one time step. We will use the parameter $\qr_N$ (``quantum read") to denote the cost of a $\mathsf{QRAM}_N$ gate, and the parameter $\qw_N$ (``quantum write") to denote the cost of a $\mathsf{QRAG}_{N}$ gate. These models are also adapted to query complexity analyses and we define a {\em query} as one application of the $\mathsf{QRAM}_N$ or the $\mathsf{QRAG}_{N}$ gate, in the respective model of computation. 

When necessary, we will explicitly state which model we are using. However, in most cases, this should be clear from the statement of the complexity result. Results
about time bounds including the parameter $\qr_N$ use the $ \mathsf{QRAM}$ model, while results including the parameter $\qw_N$ use the $\mathsf{QRAG}$
model. The $\mathsf{QRAG}_N$ gate is at least as powerful as the $ \mathsf{QRAM}_N$ gate. Indeed, it is not hard to see that one application of the latter can be simulated by two applications of the former and a constant number of one and two-qubit gates. Therefore, algorithmic results stated in the $ \mathsf{QRAM}$ model
are also valid in the $ \mathsf{QRAG}$ model with the same order of complexity.
We are not aware of an analogous reverse simulation.

\subsubsection{Inputs vs. Instances}\label{sec:instances}

The parameterization by $N$ of the memory access times $\qr_N$ and $\qw_N$ allows for quantum memories of different sizes to be accessed at different rates. We assume that the size $N$ of the memory is fixed during the running of any algorithm, and thus the cost of the memory accesses will remain constant during the algorithm, even when recursive calls are made on shorter and shorter strings.

We therefore make a distinction between an {\em instance}, which is any string on which the algorithm might be recursively called, and an {\em input}, which is the initial string the algorithm is given. We denote the size of the input by $n^*$. The size of an instance can be any integer $n \leq n^*$.

Consequently, the memory content register will be a sufficiently large function $N(n^*) \geq n^*$ of the input size, depending on the problem $P$ we are solving.
In the $\mathsf{QRAM}$ model we choose $N(n^*) = n^*$ and suppose that, at the beginning of the computation, the input $a \in \Sigma^{n^*}$ is in the memory content register. In the $\mathsf{QRAG}$ model, the size $N(n^*)$ of the memory content register can be strictly larger than $n^*$ and we suppose that, at the beginning of the computation, it contains $a{\bf 0}^{N(n^*) - n^*}$, where ${\bf 0} = 0^r$.

Thus, for any problem with input of size $n^*$, the cost of quantum read and write operations will be $\qr_{N(n^*)}$ and $\qw_{N(n^*)}$, respectively. According to some proposals for quantum random access memory implementations \cite{GLM08, AGJMS15},
$\qr_{N(n^*)}$ and $\qw_{N(n^*)}$ might scale as $O(r\log N(n^*))$, where $r$ is the number of bits each memory cell can store. 
In that case, when $N(n^*)$ is a polynomial function of $n^*$, the cost of the memory access operations would be $O(r\log(n^*))$. 
For both $\mathsf{QRAM}$ and $\mathsf{QRAG}$ models, we assume that basic arithmetic operations and comparisons can be performed in the same time as memory access.

We make two remarks. First, for notational simplicity, we will quote all final algorithmic running times in our theorems in terms of $n$ rather than $n^*$, i.e., we take the \emph{input} size to be $n$ in our final results. Second, we emphasize that our quantum running analyses include the time cost of accessing quantum memory, while when we quote classical algorithmic complexity results these assume unit cost for memory access.

\subsection{Quantum algorithms}
\label{subsec:qroutines}
We state here several basic quantum algorithms we will use in the paper.
All, except the last result on element distinctness, are in the  $\mathsf{QRAM}$ model.
\begin{fact}[{\sc Quantum search}, Grover \cite{grover1996fast} and Theorem 3 in \cite{BBHT98}]
\label{fact:search}
For $j \in [J]$, let $f_j : \Sigma^n \rightarrow \{0,1\}$ be a boolean function. Let $F$ be a quantum algorithm that for every $(j,x) \in [J] \times \Sigma^n$, 
when $\ket{x}$ is given in the quantum memory content register, correctly computes $f_j(x)$ with $q$ queries and in time~$\tau$.
The quantum search algorithm uses ${O}(\sqrt{J} )$ repetitions of $F$ and, with probability at least $9/10$ finds an index $j \in [J]$ such that $f_j(x)=1$, if there is one.
The query complexity of the algorithm is $O(\sqrt{J} q)$, and its time complexity is ${O}({\sqrt{J} ( \log J + \tau}))$.
\end{fact}

Grover's search was generalized in the query model of computation to the case where the oracle's answer is only correct with probability 8/10. The following result states that under such conditions the search is still possible with the same order of complexity as in the case of an error-free oracle.

\begin{fact}[{\sc Quantum query search with an erroneous oracle},
\cite{hoyer2003quantum}]
\label{fact:qseo}
For $j \in [J]$, let $f_j : \Sigma^n \rightarrow \{0,1\}$ be a boolean function. Let $F$ be a quantum algorithm that for every $(j,x) \in [J] \times \Sigma^n$, 
when $\ket{x}$ is given in the quantum memory content register, with $q$ queries computes $f_j(x)$ with probability at least $7/10.$
Then there exists a quantum algorithm that uses ${O}(\sqrt{J} )$ repetitions of $F$ and with probability at least $9/10$ finds an index $j \in [J]$ such that $f_j(x)=1$, if there is one.
The query complexity of the algorithm is $O(\sqrt{J} q)$.
\end{fact}

In the following corollary, we extend the above result for time complexity. This will be our main tool for the analysis of quantum divide and conquer algorithms since the results of the recursive calls might also be erroneous.
\begin{corollary}[{\sc Quantum search with an erroneous oracle}]
\label{cor:seo}
For $j \in [J]$, let $f_j : \Sigma^n \rightarrow \{0,1\}$ be a boolean function. Let $F$ be a quantum algorithm that for every $(j,x) \in [J] \times \Sigma^n$, when $\ket{x}$ is given in the quantum memory content register, in time~$\tau$ computes $f_j(x)$ with probability at least $7/10.$
Then there exists a quantum algorithm which in time~$O( {\sqrt{J} ( \log J + \tau}))$ with probability at least $9/10$ finds an index $j \in [J]$ such that $f_j(x)=1$, if there is one.
\end{corollary}

\begin{proof}
The algorithm is exactly the one given by H{\o}yer, Mosca and de Wolf~\cite{hoyer2003quantum}, we just analyze its time complexity. 
Let $m=\floor{\log_9 J}$. For $k = 1, \ldots, m$, we define recursively the unitary transformations $A_k$ which produce candidate marked elements. 
The final algorithm is, modulo some simple classical verification steps, the successive runs of $A_1, \ldots, A_m$. 
We describe these unitaries as acting on four registers, where the first register acts on the space spanned by the basis states $1, \ldots, J$, corresponding to the indices of the $J$ functions.
The second register acts on a single qubit. A qubit $\ket{1}$ in this register in principle indicates that the state in the first register is marked, but because of the evaluation errors, we can also have false positives.
The third register includes several sub-registers necessary for the full description of the algorithm, including the workspace for the evaluation of $F$, the index register for the memory access, the register for the memory content, etc.
We won't give here in detail the relatively straightforward functioning of this register and we won't indicate its content either except when it is needed
for the comprehension of the algorithm. The fourth register is the memory register that contains the input $x$, and we will not explicitly indicate
this register.
When we indicate only two registers in the description, they correspond to the first and second registers.
We set $\Gamma = \{ j \in [J] : f_j(x) = 1\}$ as the set of the marked states, and for the (unknown) cardinality $|\Gamma | =t$ we suppose that $1 \leq t \leq J/9$. Finally, we only describe the action of these unitaries on the basis state $\ket{{\bf 0}} \ket{0},$ where $\ket{{\bf 0}}$ indicates
the all $0$ state on $\ceil{\log J}$ qubits.

The initial unitary $A_1$ simply computes $F$ in superposition, that is
\[
A_1 \ket{{\bf 0}} \ket{0} = \frac{1}{\sqrt{J}} \sum_{j=1}^{J} \ket{j} \ket{f_j(x)}.
\]
We can also write this as 
\[
A_1 \ket{{\bf 0}} \ket{0} = \frac{1}{\sqrt{J}} \sum_{j}^{J} \ket{j}(\sqrt{p_j} \ket{1} + \sqrt{1 - p_j}  \ket{0}),
\]
where $p_i$ is the probability that $F$ outputs $1$ on $\ket{j} \ket{x} $.
The time complexity of $A_1$ is $\log J + \tau.$
We can express $A_1 \ket{{\bf 0}} \ket{0}$ as
$$
A_1 \ket{{\bf 0}} \ket{0} = \alpha_1 \ket{\Gamma_1}\ket{1} + \beta_1 \ket{\bar{\Gamma}_1}\ket{1} + 
\sqrt{ 1 - \alpha_1^2 - \beta_1^2 }  \ket{\phi_1}   \ket{0},
$$
where $\ket{\Gamma_1}, \ket{\bar{\Gamma}_1},  \ket{\phi_1} $ are (not necessarily uniform) superpositions of marked, non-marked, and all states, respectively. Moreover, we know that by measuring this state, we see a marked element in the first register with probability at least
$$
\alpha_1^2 = \sum_{j \in \Gamma} \frac{p_j}{J } \geq \frac{7t}{10J}.
$$

The transformation $A_{k+1}$ is defined as $A_{k+1} = - E_k A_k S_0 A_k^{-1} S_1 A_k$, whose components we now describe.
The unitary $S_0$ adds a phase $-1$ to the base state $\ket{{\bf 0}} \ket{0}$, while the unitary $S_1$ puts a phase $-1$ to every state whose second register contains $\ket{1}$. 
In other words, $-A_k S_0 A_k^{-1} S_1$ is the Grover iterate of $A_k$.
The cost of implementing $S_0$ is $O(\log J)$ while $S_1$ can be implemented in constant time. The error correction unitary $E_k$ does majority voting in superposition on $O(k)$ computations of $F$ on basis states whose second register contains $\ket{1}$. As a result of this, on such states, the error probability is exponentially reduced in $k$.
Formally, $E_k$ acts as (and here we include one qubit from the third register, which is initially $\ket{0}$):
\[
E_k \ket{j} \ket{b} \ket{0} = 
\begin{cases}
\alpha_{kj} \ket{j} \ket{1} \ket{1} + \sqrt{1 - \alpha_{kj}^2} \ket{j} \ket{0} \ket{1}  & \text{if $b=1$,} \\
\ket{j} \ket{0} \ket{0} & \text{if $b=0$,}
\end{cases}
\]
where $\alpha_{kj}^2 \geq 1 - 2^{-(k+5)}$ if $f(j) = 1$, and $\alpha_{kj}^2 \leq 2^{-(k+5)}$ if $f(j) = 0.$
The time complexity of $E_k$ is $O(k \cdot \tau).$

Similarly to $A_1 \ket{{\bf 0}} \ket{0}$, the state $A_k \ket{{\bf 0}} \ket{0}$ can be decomposed as
$$
A_k \ket{{\bf 0}} \ket{0} = \alpha_k \ket{\Gamma_k}\ket{1} + \beta_1 \ket{\bar{\Gamma}_k}\ket{1} + 
\sqrt{ 1 - \alpha_k^2 - \beta_k^2 }  \ket{\phi_k}   \ket{0},
$$
where again $\ket{\Gamma_k}, \ket{\bar{\Gamma}_k},  \ket{\phi_k} $ are superpositions of marked, non-marked, and all states, respectively. When passing from $A_k \ket{{\bf 0}} \ket{0} $ to $A_{k+1} \ket{{\bf 0}} \ket{0} $, the amplitude of correct positive states is approximately multiplied by $3$, while the amplitude of the false positives is multiplied by an exponentially small factor in $k$. The exact analysis is given in~\cite{hoyer2003quantum}], where it is proven, as a consequence, that $\alpha_{k(t)} = \Omega(1)$, for $k(t) = \floor{\log_9 (J/t)}$. 
Because $t$ is unknown, the final algorithm, for $k = 1, \ldots, \floor{ \log_9 J}$, runs $A_k$ a sufficiently large constant number of times, to amplify the success probability to close to 1.
Then it verifies each result $j$ by computing $f_j(x)$ $O(\log J)$-times and outputs a solution if one is found.

Let $T_k$ be the complexity of computing $A_k$. The overall complexity of the algorithm is then 
$T = \sum_{k=1}^m (T_k + O( \log J \cdot \tau)).$
For $T_k$, we have the following recursion: $T_1 = \log J + \tau$ and
$$
T_{k+1} = 3 T_k + O( \log J + k \cdot \tau).
$$
Considering that $ \sum_{i=1}^{k} i 3^{k-i} = O(3^k)$, we have $T_k =  {O}(3^k ( \log J + \tau))$, and the overall time complexity is 
$T = O(\sqrt{J} (\log J + \tau) + \log^2 J \cdot \tau) = O(\sqrt{J} (\log J + \tau))$.
Observe that the constant hidden in the $O$ notation doesn't depend on $F$.
\end{proof}

The next two facts essentially state that quantum minimum finding, with error-free or erroneous oracle, can be achieved with the same complexity as analogous quantum search.
\begin{fact}[{\sc Quantum minimum finding}, Theorem 1 in~\cite{DH96}]
\label{fact:minimum}
For $j \in [J]$, let $f_j : \Sigma^n \rightarrow \Z$ be a function. Let $F$ be a quantum algorithm that for every $(j,x) \in [J] \times \Sigma^n$, when $\ket{x}$ is given in the quantum memory content register, correctly computes $f_j(x)$ with $q$ queries and in time~$\tau$.
The quantum minimum finding algorithm uses ${O}(\sqrt{J} )$ repetitions of $F$ and with probability at least $9/10$ finds the index $j \in [J]$ of a minimal element in $\{f_1(x), \ldots , f_J(x) \}$.
The query complexity of the algorithm is $O(\sqrt{J} q)$, and its time complexity is ${O}({\sqrt{J} ( \log J + \tau}))$.
\end{fact}

\begin{fact}[{\sc Quantum query minimum finding with erroneous oracle}, Lemma 3.4 in\cite{WY20}]
\label{fact:qmeo}
For $j \in [J]$, let $f_j : \Sigma^n \rightarrow \Z$ be a function. Let $F$ be a quantum algorithm that for every $(j,x) \in [J] \times \Sigma^n$, when $\ket{x}$ is given in the quantum memory content register, with $q$ queries computes $f_j(x)$ with probability at least $7/10.$
Then there exists a quantum algorithm which uses ${O}(\sqrt{J} )$ repetitions of $F$ and with probability at least $9/10$ finds the index $j \in [J]$ of a minimal element in $\{f_1(x), \ldots , f_J(x) \}$.
The query complexity of the algorithm is $O(\sqrt{J} q)$.
\end{fact}

Similarly to the search case, in the next corollary we extend to the time complexity the result about minimum finding when the elements can only be evaluated with some error. \\
\begin{corollary}[{\sc Quantum minimum finding with an erroneous oracle}]
\label{cor:meo}
For $j \in [J]$, let $f_j : \Sigma^n \rightarrow \Z$ be a function. Let $F$ be a quantum algorithm that for every $(j,x) \in [J] \times \Sigma^n$, when $\ket{x}$ is given in the quantum memory content register, in time~$\tau$ computes $f_j(x)$ with probability at least $7/10.$
Then there exists a quantum algorithm which in time~$O({\sqrt{J} ( \log J + \tau}))$ with probability at least $9/10$ finds the index $j \in [J]$ of a minimal element in $\{f_1(x), \ldots , f_J(x) \}$.
\end{corollary}

\begin{proof}
Analogously to the case of erroneous quantum search, here we simply analyze the time complexity of the query algorithm of Wang and Ying of \cref{fact:qmeo}, which is essentially the same as the query algorithm of D\"urr and H{\o}yer stated in \cref{fact:minimum}. This algorithm repeatedly searches for a random element in an array that is smaller than some initially random pivot element, where in every iteration the pivot is replaced by the element recently found. 
The main difference between the two algorithms is that the search for the new pivot in the Wang and Ying algorithm uses the quantum query search algorithm with the erroneous oracle of \cref{fact:qseo}. In \cref{cor:seo} we have shown that the time complexity of this search incurs a multiplicative factor of  $( \log J + \tau)$ with respect to the number of repetitions of $F$, therefore the bounds follow from \cref{fact:qmeo}.
\end{proof}

In the applications of \cref{cor:seo} and \cref{cor:meo} it is important to have good bounds on the complexity of $F$. The following Lemma that we state in the $\mathsf{QRAM}$ model describes two simple but very generic algorithms for computing $F$ that can be used without any assumption on the functions $f_j$. 
A similar lemma can be proven in the $\mathsf{QRAG}$ model with the appropriate modification in the size of the quantum memory content register.
However, we emphasize that in many of our applications to divide and conquer algorithms we will be able to devise more efficient implementations of $F$ taking advantage of the recursive structure of the $f_j$.

\begin{lemma}
\label{lem:two}
Let $J : \N \rightarrow \N$ be a function, and for all $j \in [J(n)]$, let $f_j : \Sigma^n \rightarrow \{0,1\}$, for all $j \in [J(n)]$, let $f_j : \Sigma^n \rightarrow \Z$.
Let us suppose that for every $j \in [J(n)]$, the function $f_j(x)$ can be computed with probability at least $9/10$ by a circuit $C_j$ having $S_j(n)$ one and two-qubit gates and $q_j(n)$ query gates.  
We set $S_{\su}(n) = \sum_{j \in [J]} S_j(n), q_{\su}(n) = \sum_{j \in [J](n)} q_j(n)$, and $q_{\max}(n) = \max_{j \in [J(n)]} q_j(n)$.
Then, there exist two quantum algorithms ${\cal A}_1$ and  ${\cal A}_2$ which, with probability at least $9/10$, find the index $j \in [J(n)]$ of a marked or minimal element in $\{f_1(x), \ldots , f_J(x)\}$, respectively.  The complexities of the algorithms are as follows:
\begin{enumerate}
\item
${\cal A}_1$ makes $O(\sqrt{J(n)} q_{\max}(n)  )$ queries and takes time $O(\sqrt{J(n)} (S_{\su}(n) +q_{\max}(n) \cdot \qr_n))$,
\item
${\cal A}_2$ makes $O(\log J(n)  \cdot q_{\su}(n)  )$ queries and takes time $O(\log{J(n)} (S_{\su}(n) +  q_{\su}(n) \cdot \qr_n))$.
\end{enumerate}
\end{lemma}

\begin{proof}
By {\em tasks} we refer to both problems, search, and minimization.
The first algorithm is the application of \cref{cor:seo} and \cref{cor:meo}, with a generic method for computing $F(j,x) = f_j(x)$.
For every $j \in [J(n)]$, the circuit $C_j$ can be decomposed as
$$
U^j_1 \circ \mathsf{QRAM}_n \circ U^j_2 \circ \mathsf{QRAM}_n \circ \ldots \circ \mathsf{QRAM}_n \circ U^j_{q_j(n)+1},
$$
where the circuit $U^j_k$, for $1 \leq k \leq q_j(n)+1$, contains only one and two qubit gates. 
In fact, we can suppose without loss of generality that $q_j(n)= q_{\max}(n),$ for every $j \in [J(n)]$. Indeed the circuit
$$
C_j \circ \mathsf{QRAM}_n \circ \id \circ  \mathsf{QRAM}_n  \ldots \circ \mathsf{QRAM}_n \circ \id,
$$
where the number of appended $\mathsf{QRAM}_n$ gates is $q_{\max}(n)  - q_j(n)$, computes the same functions as $C_j$ whenever $q_{\max}(n)  - q_j(n)$ is even since $\mathsf{QRAM}_n$ is its own inverse. Observe that we increased only the number of $\mathsf{QRAM}_n$ gates, the number of one and two-qubit gates didn't change.

We use a control register with $J(n)$ qubits, where the integer $j \in [J(n)]$ is expressed in unary, that is by the binary vector whose all but the $j$th coordinates are $0$.
For every $j,k$, define the circuit $\cont{U^j_k}$ as $U_k$ whose gates are controlled by the $j$th bit of the control register. Since there exists a constant $\gamma$ such that every controlled one and two-qubit gate can be expressed with $\gamma$ one and two-qubit gates, $\cont{U^j_k}$ can be implemented by a circuit whose number of one and two-qubit gates is at most $\gamma$-times the number of one and two-qubit gates in $U^j_k$.
Then the circuit
$$
F' = \cont{U^1_1} \circ \ldots \circ  \cont{U^J_1} \circ \mathsf{QRAM}_n \circ \ldots \circ \mathsf{QRAM}_n \circ \cont{U^1_{q+1}} 
\circ \ldots  \circ  \cont{U^J_{q+1}},
$$
is implementable with $O(S_{\su}(n))$ one and two qubit gates and with $q_{\max}(n)$ query gates. 
When the control register contains $j$ in unary, $F'$ computes the function $f_j$. By definition, the algorithm $F$ on $\ket{j} \ket{x} $ transforms the binary index $j$ into unary and writes it into the control register, then it executes $F'$ and un-computes the control register. This implementation of $F$ uses $q_{\max}(n)$ query gates and time $O(S_{\su}(n) + q_{\max}(n) \cdot \qr_n)$.
Therefore by \cref{fact:qseo} and \cref{cor:seo} (respectively \cref{fact:qmeo} and \cref{cor:meo}) the tasks can be computed with ${O}({\sqrt{J(n)} q_{\max}(n)})$ queries and in time $O(\sqrt{J(n)} ( \log J(n) +   S_{\su}(n) + q_{\max}(n) \cdot \qr_n)).$ 
Since the $\log J(n)$ term is negligible, the claim follows.

The second algorithm repeats $C_j$ $ O(\log J(n))$-times, for every $j \in [J(n)]$, to reduce its error probability to $1/(10J(n))$. Then it classically computes the minimum of the $J(n)$ values. The complexities follow immediately.
\end{proof}

\begin{fact}[{\sc Element distinctness} Ambainis~\cite{Amb07}]
\label{fact:ed}
Let $a$ be a string of $n$ of characters. 
With probability at least $9/10$,  the quantum element distinctness algorithm decides if all the characters are distinct.
If we have oracle access to $a$ then the algorithm makes ${O}(n^{2/3} )$ queries, and the time complexity of the algorithm is~$\widetilde{O}(n^{2/3}) \cdot \qw_{O(n)}$.
\end{fact}

\section{Bottom-up divide and conquer}
\label{sec:bottomup}
When the create function is trivial then the subproblems that arise in a divide and conquer algorithm are known in advance.  In this case, we can directly write an iterative algorithm rather than a recursive one, by sequentially solving the subproblems in an appropriate order.  A classic example is mergesort.
Here the create function is trivial as the subproblems are to sort the left and right halves of the input string.  In bottom-up mergesort, we start at the ``bottom'' by sorting every two-element substring in positions $2i-1, 2i$.
Then we sort 4-element substrings in positions $4i-3, 4i-2,4i-1,4i$ by merging the already sorted 2-element substrings, and so on, working up to the full string.  

In the next subsection, we develop a generic bound on the running time of a bottom-up quantum algorithm.  In the following subsection, we then look at several applications of this framework.

\subsection{Generic bottom-up bound}
In the problems we look at in this paper, we are frequently trying to find a substring of a string $a \in \Sigma^n$ which maximizes some function. 
Examples of problems of this type include computing the length of a longest substring with no repeating characters, finding the length of a longest substring consisting 
of only a single character, or for a string of integers maximizing the difference between the last and first elements of a substring.

To see how the bottom-up approach to such a problem works, consider the endpoints $i,j \in [n]$ of a substring which maximize the function of interest.  When viewed as bit strings of length $\log n$ corresponding to the binary representation of $i-1, j-1$ respectively, these strings will share a common prefix of length $t \in \{0, 1, \ldots, \log(n) -1\}$.
This means that $i,j$ are contained in a unique interval of length $n/2^t$ of the form $\{(k-1)n/2^t+1, \ldots, kn/2^t\}$, for some $k \in [2^t]$, such that $i$ is in the left half of this interval, and $j$ is in the right half of this interval.  

Thus to solve the original problem it suffices to solve the \emph{crossing problem}: compute the maximum value $P(a, t, k)$ of the function of interest on substrings contained in an interval $\{(k-1)n/2^t+1, \ldots, kn/2^t\}$ specified by $k,t$ whose left endpoint is in the left half of the interval and the right endpoint is in the right half of the interval. 
The solution to the original problem is then $\max_t \max_{1 \le k \le 2^t} P(a, t, k)$.

The next theorem gives an upper bound on the quantum complexity of the bottom-up approach in terms of the quantum complexity of solving the crossing problem $P(a, t, k)$. 
Note that the problem $P(a, t, k)$ is a function of a string of size $n/2^t$; in the next theorem, we suppose that we have a bounded-error quantum algorithm for $P(a, t, k)$ whose running time is $O(\sqrt{n/2^t} \cdot \tau(n, t))$.  This parameterization is taken as we will only encounter quantum algorithms whose running time is at least quadratic in the input length, and explicitly factoring out the square root term gives a more elegant expression.
\begin{theorem}
\label{thm:bu_easy}
Let $\Sigma$ be an alphabet, $n$ a power of 2, and $T = \{0, \ldots, \log(n) - 1\}$.
Let $U = \Sigma^n \times T \times [n]$, and $S = \{(a, t, k) \in U : k \le 2^t\}$.
Let $P : S \rightarrow \Z$ be a function. 
Suppose that for every $t \in T$ there is a quantum algorithm $A_t$ that for every $a \in \Sigma^n, 1 \le k \le 2^t$ outputs $P(a, t, k)$ with probability at least $9/10$ in time $\sqrt{n/2^t} \cdot \tau(n, t)$, and with $\sqrt{n/2^t} \cdot \sigma(n, t)$ queries, for some functions 
$\tau, \sigma :\N \times \N \rightarrow \R_+$.
Then there is a quantum algorithm that computes 
\[
\max_{t, 1 \le k \le 2^t} P(a, t, k)
\]
and the $t,k$ realizing the maximum with probability at least $9/10$ in time 
\[
O(\sqrt{n} \log \log(n) \cdot (\log n + \sum_{t \in T} \tau(n, t)))
\]
and with
\[
O(\sqrt{n} \log \log(n) \cdot \sum_{t \in T} \sigma(n, t))
\]
queries.
\end{theorem}

\begin{proof}
For any $a \in \Sigma^n, t \in T$, we can compute $\argmax_k P(a,t,k)$ with probability at least $9/10$ in time $O(t\sqrt{2^t} + \sqrt{n}\tau(n, t))$ by \cref{cor:meo} and with 
$O(\sqrt{n} \sigma(n, t))$ queries by \cref{fact:qmeo}. 
For the $k^*$ realizing the maximum we then compute the value $P(a, t, k^*)$.
For each $t \in T$, we repeat this procedure $32 \log \log n$ times and take the majority answer, which will be equal to $\max_{1 \le k \le 2^t} P(a, t, k)$ with probability at least
$1 - 1/(10 \log(n))$. We then take the maximum value over all $t \in T$ which is equal to $\max_t \max_{1 \le k \le 2^t} P(a, t, k)$ with probability at least $9/10$ by the union bound. 
Summing the complexity bounds over $t \in T$ gives the theorem.
\end{proof}

\cref{thm:bu_easy} is a simple upper bound on the quantum complexity of a quantum algorithm using the bottom-up method, and is good to use when we do not care about optimizing logarithmic factors. In the rest of this section, we focus on problems that have a quantum running time around $\sqrt{n}$, for which we give a tailored version of \cref{thm:bu_easy}.

Recall that, for $k \geq 1$, we defined the function
\[
\lambda_k (n, m) = \min \{\log^k n + m, \log^{(k+1)/2} (n) (\log \log n)^{k-1} + \log^{(k-1)/2} (n) (\log \log n)^{k-1} \cdot m \}.
\]
Observe that  $\lambda_1 (n, \qr_n) =  \log n + \qr_n,$ and that $O(\sqrt{n} \cdot \lambda_1 (n, \qr_n) )$ is an upper bound on the time complexity of Grover's search.
The function 
\[
\lambda_2 (n, \qr_n) = \min \{\log^2 n + \qr_n , \log^{3/2} (n) \log \log n + \sqrt{\log n} \log \log n \cdot \qr_n \}
\]
arises in the next theorem with application to the \SSSTf \ and related problems.
\begin{theorem}
\label{thm:bu_rootn}
Let $\Sigma$ be an alphabet, $n$ a power of 2, and $T = \{0, \ldots, \log(n) - 1\}$.
Let $U = \Sigma^n \times T \times [n]$, and $S = \{(a, t, k) \in U : k \le 2^t\}$.
Let $P : S \rightarrow \Z$ be a function.  Suppose that for every $t \in T$, there is a quantum algorithm $A_t$ that for every $a \in \Sigma^n, 1 \le k \le 2^t$ outputs $P(a, t, k)$ with probability at least $9/10$ in time $\sqrt{n/2^t} \cdot (\log(n) + \qr_n)$, and with $O(\sqrt{n/2^t})$ queries.  Then there is a quantum algorithm that computes 
\[
\max_t \max_{1 \le k \le 2^t} P(a, t, k)
\]
and the $t,k$ realizing the maximum with probability at least $9/10$ with $O(\sqrt{n \log(n)})$ queries, and a quantum algorithm to do this with running time $O(\sqrt{n \log(n)} \lambda_2(n, \qr_n))$.
\end{theorem}

\begin{proof}
For any $a \in \Sigma^n, t \in T$, we can compute $\max_{1 \le k \le 2^t} P(a,t,k)$ with probability at least $9/10$ in time $O(t\sqrt{2^t} + \sqrt{n}(\log(n) + \qr_n))$ by \cref{cor:meo} and with $O(\sqrt{n})$ queries by \cref{fact:qmeo}.
To compute $\max_t \max_{1 \le k \le 2^t} P(a, t, k)$ we use \cref{lem:two}.
The first item of the lemma gives an algorithm with $O(\sqrt{n \log n})$ queries and running time $O(\sqrt{n \log(n)}(\log^2(n) + \qr_n))$. 
The second item of the lemma gives an algorithm with running time $O(\sqrt{n} \log \log(n) \cdot (\log^2(n) + \log(n) \qr_n))$. 
Taking the minimum of the two options for the time complexity gives the claimed bound of $O(\sqrt{n \log(n)} \lambda_2(n, \qr_n))$.
\end{proof}

\subsection{Applications}
We can directly apply \cref{thm:bu_rootn} to the following problems:

\begin{problem}[\SSSTf]
Given a length $n$ array $A$ of integers, compute 
\[
\argmax_{1 \leq i < j \leq n} A_j - A_i \enspace .
\]
\end{problem}

\begin{problem}[\LISstf \ (\LISst)]
Given an array $A$ consisting of $n$ integers, find $ i < j $ such that $A_i < A_{i+1} \ldots < A_j$ is the longest increasing substring, if any.
\end{problem}

\begin{problem}[\LSICf \ (\LSIC)]
Given a string $a \in \Sigma^n$ for some finite alphabet $\Sigma$, find $i<j$ such that $a_i = a_{i+1} = \ldots = a_j$ is the longest substring of identical characters, if any.
\end{problem}

\begin{problem}[\LZSf \ (\LZS)]
Given a string $a \in \Sigma^n$, for $\Sigma = \{0,1,2\}$, find $i <j$ such that $a[i:j] = 20^{j-i-1}2$ is a longest substring from $20^*2$, if any.
\end{problem}

\LZS \ is a slight generalization of {\Recognizing} $\Sigma^{*}20^*2\Sigma^{*}$, the problem to decide if $a \in \Sigma^{n}$ is in the regular language $\Sigma^{*}20^*2\Sigma^{*}$.
In ~\cite{AGS19, CKKSW22} it was proven, using different techniques, that the query complexity of {\Recognizing} $\Sigma^{*}20^*2\Sigma^{*}$ is $O(\sqrt{n \log n})$.

\begin{theorem}
\label{thm:three}
The quantum query and time complexities of the problems \SSST, \LISst, \LSIC \ and \LZS \ are respectively $O(\sqrt{n \log n})$ and $O({\sqrt{n \log n} \cdot \lambda_2 (n, \qr_n}))$.
\end{theorem}

\begin{proof}
Let $n$ be a power of 2, and for $t \in \{0, \ldots, \log(n) -1\}, 1\le k \le 2^t$ define  $C_{n, t, k} = \{(k-1)n/2^t + 1, \ldots, kn/2^t\}$. Note that $C_{n, t, k} = C_{n, t+1, 2k-1} \cup C_{n, t+1, 2k}$.

First, consider the problem \SSST.  Let $a$ be an array of integers of size $n$.  Define the problem $P(a, t, k)$ to be
\[
P(a, t, k) = \max_{i \in C_{n, t+1, 2k-1} \atop j \in C_{n, t+1, 2k}} a[j] - a[i] \enspace.
\]
For any $i < j \in [n]$ there is a $t, k$ such that $i \in C_{n, t+1, 2k-1}$ and $j \in C_{n, t+1, 2k}$.  Thus 
\[
\max_{i < j} a[j] - a[i] = \max_t \max_{1 \le k \le 2^t} P(a, t, k) \enspace .
\]
By \cref{fact:minimum} we can compute $P(a, t, k)$ with $O(\sqrt{n/2^t})$ queries and in time $O(\sqrt{n/2^t} (\log n + \qr_{n/2^t}))$.  Thus by \cref{thm:bu_rootn} there is a quantum algorithm to compute $\max_t\max_{1 \le k \le 2^t} P(a, t, k)$ and the $t,k$ realizing this value with probability at least $9/10$ using $O(\sqrt{n\log(n)})$ queries and a quantum algorithm to do this with time $O(\sqrt{n\log(n)}\lambda_2(n, \qr_n))$.
After we find the $t,k$ realizing the maximum we can again solve $P(a, t, k)$ to find $\argmax_{i < j} a[j] - a[i]$ and output these values.

We now prove by induction that there is a quantum algorithm for any $n$ with the desired complexity.  The base case is $n=2$ which is a power of 2 and so done.
Now we design an algorithm for an input of $a$ size $2^m < n < 2^{m+1}$ assuming we have an algorithm of the desired complexity for inputs of size at most $2^m$.  
There are three choices for $\argmax_{i < j} a[j] - a[i]$: either $1 \le i,j \le 2^m$, $2^m < i,j \le n$, or $i \le 2^m, j > 2^m$.  The first two cases can be solved by the inductive hypothesis, considering $a[1 \upto 2^m]$ and $a[2^m +1 \upto n]$, respectively.  For the third case, we can pad $a$ to an array $a'$ of size $2^{m+1}$ by adding $-\infty$ entries to the end and then solve the problem $P(a', 0, 1)$.  We then take the maximum of the three cases.  The running time is dominated by the first two cases and so is as desired.

Next, consider \LISst.  This problem can be trivially padded by repeating the last element, thus we may assume the input size is a power of 2.  
Let the problem $P(a, t, k)$ be defined as the length of a longest increasing substring of $a[(k-1)n/2^t + 1 \upto kn/2^t]$ that includes both of $a[(2k-1)n/2^{t+1}]$ and $a[(2k-1)n/2^{t+1}+1]$. 
We again have that the length of a longest increasing substring of $a$ is equal to $\max_t \max_{1 \le k \le 2^t} P(a, t, k)$.
Thus by \cref{thm:bu_rootn} it suffices to show that we can compute $P(a, t, k)$ with $O(\sqrt{n/2^t})$ queries and in time $O(\sqrt{n/2^t} (\log n + \qr_{n/2^t}))$.  
To compute $P(a, t, k)$ we first compute
\begin{align*}
i^* &= 
\begin{cases}
(k-1)n/2^t + 1 & \text{ if } a[i] < a[i+1] \text{ for all } i \in C_{n, t+1, 2k-1} \\
1 +\argmax_{i \in C_{n, t+1, 2k-1}} a[i] \ge a[i+1] & \text{ otherwise }
\end{cases} \\
j^* &= 
\begin{cases} 
kn/2^t & \text{ if } a[i-1] < a[i] \text{ for all } i \in C_{n, t+1, 2k} \\
-1 + \argmin_{j \in C_{n, t+1, 2k}} a[j] \le a[j-1] & \text{ otherwise } \enspace.
\end{cases} 
\end{align*}
These can be computed with $O(\sqrt{n/2^t})$ queries and in time $O(\sqrt{n/2^t} (\log n + \qr_{n/2^t}))$ by \cref{fact:minimum}.  The result follows by noting that $P(a, t, k) = j^* - i^* + 1$.

For \LSIC, we can pad by repeating two distinct characters at the end of the string and thus we may again assume that $n$ is a power of 2.  
The problem $P(a, t, k)$ is to compute the length of a longest substring of identical characters in of $a[(k-1)n/2^t + 1 \upto kn/2^t]$ that includes both of $a[(2k-1)n/2^{t+1}]$ and $a[(2k-1)n/2^{t+1}+1]$.
We can solve this in a very similar way to \LISst. Let
\begin{align*}
i^* &= 
\begin{cases}
(k-1)n/2^t + 1 & \text{ if } a[i] = a[i+1] \text{ for all } i \in C_{n, t+1, 2k-1} \\
1 +\argmax_{i \in C_{n, t+1, 2k-1}} a[i] \ne a[i+1] & \text{ otherwise }
\end{cases} \\
j^* &= 
\begin{cases} 
kn/2^t & \text{ if } a[i-1] = a[i] \text{ for all } i \in C_{n, t+1, 2k} \\
-1 + \argmin_{j \in C_{n, t+1, 2k}} a[j] \ne a[j-1] & \text{ otherwise } \enspace.
\end{cases} 
\end{align*}
Then again $P(a, t, k) = j^* - i^* + 1$ and can be computed with $O(\sqrt{n/2^t})$ queries and in time $O(\sqrt{n/2^t} (\log n + \qr_{n/2^t}))$ by \cref{fact:minimum}.  

Finally, we turn to \LZS.  We may assume the input size is a power of 2 by padding the end of the input with zeros. Let $P(a, t, k)$ be the length of a longest $20^*2$ substring of $a$ whose left endpoint is in $C_{n, t+1, 2k-1}$ and whose right input is in $C_{n, t+1, 2k}$.  To compute $P(a, t, k)$ we compute 
\begin{align*}
i^* &= 
\begin{cases}
- \infty & \text{ if } \{i \in C_{n, t+1, 2k-1} : a[i] = 2\} = \emptyset \\
\max \{i \in C_{n, t+1, 2k-1} : a[i] = 2\} & \text{ otherwise.} 
\end{cases} \\
j^* &= 
\begin{cases}
\infty & \text{ if } \{i \in C_{n, t+1, 2k} : a[i] = 2\} = \emptyset \\
\min \{i \in C_{n, t+1, 2k} : a[i] = 2\} & \text{ otherwise.}
\end{cases}
\end{align*}
This can be done with $O(\sqrt{n/2^t})$ queries and in time $O(\sqrt{n/2^t} (\log n + \qr_{n/2^t}))$ by \cref{fact:minimum}. If 
either $i^* = -\infty$ or $j^* = \infty$ then $P(a, t, k) = 0$. Otherwise, we use Grover search (\cref{fact:search}) to check that the interval $\{i^*+1, \ldots, j^*-1\}$ contains only zeros, which takes the same time and queries asymptotically. 
If so, $P(a, t, k) = j^* - i^* + 1$; otherwise, $P(a, t, k) = 0$
\end{proof}

\subsection{Generalizations}
There are several interesting generalizations of the paradigmatic \SSST \ problem. The first one we consider is the $d$-\MSSTf \  ($d$-\MSST) problem, where $d \geq 1$ is a constant. Here, we are given $d \geq 1$ stocks, and for each stock, we make a single transaction. Every stock should be bought before it is sold, and the aim is to maximize the total profit. 
If we have as input $d$ arrays of length $n$, each array specifying the prices for one of the stocks on $n$ successive dates, then we can independently execute the one-dimensional algorithm $d$-times, once for each stock.
However, we want to address a more general pricing model.  
As input, we will have a $d$-dimensional array $A$, where for every $i \in [n]^d$, the integer $A_i$ is the total price of all $d$ stocks, where the transaction time for the $k$th stock is $i_k$, for $k = 1, \ldots, d$. 
In more abstract terms, we want to find $\max_{i < j} A_j - A_i$ when we have the natural strict partial order over the $d$-dimensional cube of side $n$.

\begin{problem}[$d$-\MSSTf]
Given a $d$-dimensional cube $a$ of side $n$ of integers,  find $\arg\max_{i < j} A_j - A_i$ where, by definition, $i<j$ when $i_k < j_k$ for $1 \leq k \leq d.$
\end{problem}

Our algorithm for $d$-\MSST \ is a direct generalization of the one-dimensional case.
\begin{theorem}
\label{thm:ssmt}
For every $d \geq 1$, the quantum query complexity of $d$-\MSST \ is $O( (n \log n)^{d/2})$ and its time complexity is
$$
O \left ((n \log n)^{d/2} \cdot \min \{\log^{d+1} n + \qr_{n^d} , \log^{1 + d/2} (n) \log \log n + \log^{d/2} (n) \log \log n \cdot \qr_{n^d} \} \right).
$$
\end{theorem}
\begin{proof}
The proof is similar to the one-dimensional case, and we detail it for the query complexity.
Again, we assume that $n$ is a power of $2$ and, for an index $i \in [n]^d$, we identify each of its coordinates $i_k \in [n]$ with the binary representation of $i_k -1$.
Consider the indices $i < j$ which realize the maximum. For every $1 \leq k \leq d,$ the strings $i_k$ and $j_k$ share some 
prefix of length $t_k \in \{0,1,\ldots, \log n-1\}$.   
This means that $i$ and $j$ are  contained in a box ($d$-dimensional sub-rectangle) $B$ of volume $n^d / 2^{\sum_{k=1}^d t_k}.$
Moreover, let us consider the $2^d$ sub-boxes of $B$ of sizes $2^{-d} |B|$ which are obtained by halving each side of $B$.
Then $i$ is in the sub-box $B_0$ which is obtained as the product of the left halves of the sides, and $j$ is in the sub-box $B_1$ obtained as the product of the right halves. By definition, every index in $B_0$ is smaller than any index in $B_1$.
Thus we can compute an optimal pair $i,j \in B$ by computing the minimum of the elements of $a$ with indices in $B_0$ and the maximum of the elements of $a$ with indices in $B_1$. This can be done  with $O( \sqrt{ n^d / 2^{\sum_{k=1}^d t_k} } )$ queries
and in time $O( \sqrt{ n^d / 2^{\sum_{k=1}^d t_k} } ) (\log n + \qr_s)$, where $s =  n^d / 2^{\sum_{k=1}^d t_k} $.

Let us consider any prefix lengths sequence $(t_1, \ldots , t_d) \in \{0,1,\ldots, m-1\}^d$. The number of corresponding boxes is $ 2^{\sum_{k=1}^d t_k}$. An optimal pair $(i,j)$ over all these boxes can be therefore found with $O( \sqrt{2^{\sum_{k=1}^d t_k}}   \sqrt{ n^d / 2^{\sum_{k=1}^d t_k} }) = O(n^{d/2})$ queries and in time $O(n^{d/2} (\log n + \qr_{n^d}))$.

To finish the proof, we use \cref{lem:two} with the $N = \log^d n$ possible values for the size of the shared prefix $t$ and with $S_t(n) = O({n}^{d/2}  \log n )$ and $q_t(n) = O({n}^{d/2})$, for all $t \in [\log^d n]$. For the query complexity, we use the first algorithm from \cref{lem:two}, while for the time complexity, we take the minimum of the two algorithms.
\end{proof}

Two other generalizations of the \SSST \ problem, $k$-\SSMTf \ and $\MaxFourCf$ will be addressed respectively in \cref{sec:is} and \cref{sec:apsp}.

\section{The $k$-Increasing Subsequence and related problems}
\label{sec:is}
The results of this section are proven in the $\mathsf{QRAM}$ model.
In the $k$-\ISf \ problem ($k$-\IS), where $k \geq 2$, we are given an array of $n$ integers and we are looking for a subsequence of $k$ increasing numbers.
\begin{problem}[$k$-\ISf]
Given an array $A$ of $n$ integers, do there exist $k$ indices $i_1 < \cdots < i_k$ such that $A_{i_1} < \cdots < 
A_{i_k}$?
\end{problem} 
For $k \geq 1$ and for an integer $\gamma \in [(\min_{1 \leq i \leq n} A_i) \ldots (\max_{1 \leq i \leq n}  A_i)] \cup \{-\infty\}$, we consider the following helper function:
\[
F_k(A,\gamma) = \min_{i_1 < i_2 < \cdots < i_k, \atop \gamma < A_{i_1} < A_{i_2} < \cdots < A_{i_k}} A_{i_k},
\]
where we adhere to the convention that the value of the minimum taken over the empty set is $\infty$.
It suffices to design an algorithm for $F_k(A, \gamma)$ to solve $k$-\IS \ with the same complexity, and in the following theorem, we do exactly that. 
We will also give explicit bounds on the constant involved in the query complexity as a 
function of $k$.
Let $\alpha$ be the universal 
constant from \cref{fact:minimum} and ~\cref{fact:qmeo} such that quantum minimum finding on an array of size $n$ can be solved with at most $\alpha \sqrt{n}$ quantum queries.  

\begin{theorem}
\label{thm:k-is}
There exists a quantum algorithm that evaluates $F_k$ on an array $A$ of $n$ integers and for any $\gamma$, with at most 
$(2\alpha)^{2k} \sqrt{n \log^{k-1} n}$ queries. 
There is also an algorithm that solves the problem in time $O(\sqrt{n \log^{k-1} n} \cdot  \lambda_k(n,\qr_n))$.
\end{theorem}

\begin{proof}
Let $T_{k}(n)$ be the quantum query complexity of $F_k(A, \gamma)$ on size $n$ strings, maximized over all values of $\gamma$. 
We prove by induction on $k$ that 
\[
T_{k}(n) \le (2\alpha)^{2k} \sqrt{n \log^{k-1} n} \enspace.
\]
The case $k=1$ follows by \cref{fact:minimum}.

Let us assume that we have algorithms for $F_1, \ldots, F_{k-1}$ with the desired quantum query complexity. 
We now design an algorithm for $F_k(A, \gamma)$. 
Consider an array $A$ of size $n$ and suppose that it has a $k$-\IS \ all of whose values are more than $\gamma$.
Let $A_{i_1} <  \cdots < A_{i_k}$ be such a subsequence where $i_k$ is such that $A_{i_k}$ achieves the minimum value possible 
among all $k$-\IS\  with the condition $A_{i_1} > \gamma$. 

Again we suppose that $n$ is a power of $2$, and we identify an index $i \in [n]$ with the binary representation of $i -1$.
Let $t \in \{0, \ldots, \log n-1\}$ be the size of the longest common prefix of $i_1$ and $i_k$. Then there exists a unique interval of length $n/2^t$ of the form $\{(k-1)n/2^t+1, \ldots ,  kn/2^t\}$, for some $k \in [2^t]$, such that the indices $i_1 , \ldots , i_k$ are all contained in it, and moreover $i_1$ is in the left half of this interval, and $i_k$ is in the right half of this interval.  

Let us first consider how to solve this problem when we know that interval
(or equivalently the largest common prefix of $i_1$ and $i_k$).
Let $C$ be the subarray of $A$ whose indices are in the interval, and let $C_\ell$ be its left half, and $C_r$ be its right half.  
Now note that 
\[
F_k(C, \gamma) = 
\min_{j \in [k-1]} F_{k-j}(C_r, F_j(C_\ell, \gamma))
\enspace .
\]
Thus $F_k(C,\gamma)$ can be expressed in terms of $F_j$, for $j < k$, and its complexity is upper bounded by $2\sum_{j=1}^{k-1} T_{ j}(|C|)$.   

Now let us return to the original problem. Say that, instead of being told the longest shared prefix of $i_1$ and $i_k$, we are only told that $i_1, i_k$ share a prefix of size $t$.  
There are $2^t$ prefixes of size $t$. We want to find the minimum of $F_k(C,\gamma)$ over all substrings $C$ defined by these prefixes, where each interval is of size $n/2^t$.  Using quantum minimum  finding over all prefixes of size $t$, by \cref{fact:qmeo} the complexity of this is at most $P_{t, k-1}(n)$ where
\begin{align*}
P_{t,k-1}(n) &= 
2 \alpha \sqrt{2^t} \sum_{j=1 }^{k-1} T_{ j} (n/2^{t}) \\
&\le 2 \alpha \sqrt{n} \sum_{j=1}^{k-1} (2\alpha)^{2j} \sqrt{\log^{j-1} (n/2^t)} \\
&= 2\alpha \sqrt{n} (2\alpha)^2 \frac{(2\alpha)^{2(k-1)} 
\log^{(k-1)/2} (n/2^t)-1}{(2\alpha)^2 \sqrt{\log (n/2^t)} - 1} \\
&\le 4 \alpha (2\alpha)^{2(k-1)} \sqrt{n \log^{k-2} n}
\enspace ,
\end{align*}
where to arrive at the last inequality we first simplify the denominator using the fact that $a-1 \ge a/2$ for $a \ge 2.$
Note that this final upper bound on $P_{t, k-1}(n)$ is independent of $t$.

The final value of $F_k(a, \gamma)$ is the minimum over all possible sizes of the shared prefix.  
Again by \cref{fact:qmeo}, this shows that $T_{k}(n)$ is at most
\begin{align*}
T_{k}(n) &\le \alpha \sqrt{\log n} P_{0, k-1}(n) \\
&\le (2\alpha)^{2k} \sqrt{n \log^{k-1} n} \enspace,
\end{align*}
as desired.
Let us now turn to the time complexity.
Let $\overline{T}_{k}(n)$ be the maximum quantum time complexity of $F_k(A,\gamma)$ over all arrays $A$ of size $n$ and all $\gamma$.
By \cref{fact:minimum} we know that $\overline{T}_{1}(n) ={O}({\sqrt{n} ( \log n + \qr_n})).$ For $k \geq 2,$ we will design two algorithms for the problem, and the claim of the Theorem will follow by taking the minimum of the two complexities.
They both proceed by induction on $k$ and they follow the same structure as the algorithm for the query complexity until the last step. Then they call on \cref{lem:two} where they use the first and second algorithms there, respectively.

We claim that when the first algorithm is chosen in \cref{lem:two}, its time complexity $\overline{T}^{(1)}_{k}(n)$ satisfies $\overline{T}_{k}^{(1)}(n) = O(\sqrt{n \log^{k-1} n} \cdot ( \log^k n +  \qr_n))$.
Computing $F_k(A,\gamma)$ under the condition that $i_1$ and $i_k$ share a prefix of size $t$, and $i_1$ is in the left half and $i_k$ is in the right half of the corresponding interval of size $n/2^t$, by \cref{cor:meo} takes time
\[
\overline{P}^{(1)}_{t,k-1} (n) 
= O \left(\sqrt{2^t} \left(t+ \sum_{j=1 }^{k-1} \overline{T}^{(1)}_{j}(n/2^t )\right) \right) . 
\]
From this, using the inductive hypothesis for $\overline{T}^{(1)}_{j}$, for $j \in [k-1]$, we get 
$$
\overline{P}^{(1)}_{t,k-1}(n)  = O(\overline{T}^{(1)}_{k-1}(n)).
$$
To maximize over all values of $t$ we use the first algorithm of \cref{lem:two} with the functions
$S_{t,k-1}(n) = O( \sqrt{n \log^{k-2} n}  \log^{k-1} n)$, and $q_{t,k-1}(n) =  \sqrt{n \log^{k-2} n} $, for all $t$.
Therefore we have $S_{\su,k-1}(n) = O( \sqrt{n \log^{k-2} n}  \log^{k} n)$ and $q_{\max,k-1}(n) =  \sqrt{n \log^{k-2} n} $,
and the Lemma gives
\[
\overline{T}_{k}^{(1)}(n) = O\left(\sqrt{n \log^{k-1} n} \cdot ( \log^k n +  \qr_n)\right).
\]

The second algorithm is the same until the use of \cref{lem:two} to maximize over all values of $t$, where it uses the second algorithm of the Lemma. It is not hard to check that the time complexity of this algorithm is
\[
\overline{T}_{k}^{(2)}(n) = O\left(\sqrt{n \log^{k-1} n} \cdot \left( \log^{(k+1)/2} (n) (\log \log n)^{k-1} + \log^{(k-1)/2} (n) (\log \log n)^{k-1} \cdot \qr_n \right)\right).
\]
Taking $\overline{T}_{k}(n) = \min \{\overline{T}_{k}^{(1)}(n) , \overline{T}_{k}^{(2)}(n) \}$,
finishes the proof.  
\end{proof}

We now turn to the second generalization of the \SSST \ problem.
For a constant $k \geq 1,$ in the $k$-\SSMTf \ problem ($k$-\SSMT) we have the same input as in the case of  \SSST \ but we buy and sell the stock $k$-times in the given time interval. 
We want to maximize the profit where the restriction on the timing is that $i$th  buying day must precede the $i^\text{th}$ selling day which, in turn, must precede the $(i+1)^\text{st}$ buying day.
We define the variant where the output is the maximal profit.

\begin{problem}[$k$-\SSMTf] 
Given a length $n$ array $A$ of integers, maximize $\sum_{\ell =1}^k A_{j_{\ell}} - A_{i_{\ell}}$ under the condition $i_1 < j_1 < i_2< \ldots < j_k$.
\end{problem}

Our claim is that in an array of size $n$ this problem can be solved by a quantum algorithm with $O(\sqrt{n \log^{2k-1} n})$
queries and in time $O(\sqrt{n \log^{2k - 1} n} \cdot \lambda_{2k} (n, \qr_n))$. 
We can directly give an algorithm for this claim. However, it turns out that it is simpler and more elegant to give an algorithm for a more general problem where, again, $k \geq 1$ is a constant.

\begin{problem}[$k$-\SiSf] 
Given a length $n$ array $A$ of integers and $\varepsilon \in \{-1,1\}^k$, maximize $\sum_{m =1}^k  \varepsilon_m A_{i_m}$ when the $k$ indices must satisfy
$i_1  < i_2< \ldots < i_k$. 
\end{problem}

Clearly computing the value of a solution of the $k$-\SSMTf \ problem is the special instance of the $2k$-\SiSf \ problem with $\varepsilon = (-1,1, \ldots, -1,1)$.
The complexity of $k$-\SiSf\ ($k$-\SiS) \ depends on $\varepsilon$. For example, for $\varepsilon = 1^k$, it can be solved by
$k$ repeated maximum findings. Our next result bounds the quantum complexity of the problem.

\begin{theorem}
\label{thm:k-sis}
Let $\alpha$ be the universal constant from \cref{fact:minimum}. 
Then there is a quantum algorithm that solves the $k$-\SiS \ problem in any array $A$ of $n$ integers and for any $\varepsilon \in \{-1,1\}^k$ with at most $$(2\alpha)^{2k} \sqrt{n \log^{k-1} n}$$ queries and in time $O(\sqrt{n \log^{k-1} n} \cdot \lambda_k(n,\qr_n))$.
\end{theorem}

\begin{proof}
Let $F_k$ be the function defined as
\[
F_k(A, \varepsilon) =  \max_{i_1 < \ldots < i_k } \sum_{j =1}^k  \varepsilon_j A_{i_j}
 \enspace.
\] 
Let $T_{n,k}$ and $\overline{T}_{n,k}$  be respectively the quantum query and time complexities of $F_k(A, \varepsilon)$ when maximized over all arrays $A$ of size $n$ and over all $\varepsilon \in \{-1,1\}^k$. 
We claim the same bounds for $T_{n,k}$ and $\overline{T}_{n,k}$ as the bounds obtained for the similarly denoted functions in \cref{thm:k-is} which were the respective complexities of $F_k(A, \gamma)$ there.
The proof is by induction on $k$, and the case $k=1$ is just quantum maximum finding.

The proofs of the inductive steps are almost identical to the ones in \cref{thm:k-is}, the only difference is in the recursive combinatorial statement for a substring $C$ of $A$ that fully contains an optimal solution $A_{i_1}, \ldots, A_{i_k}$, for $i_1 < \ldots < i_k$, and moreover is such that $A_{i_1}$ is in the left half of $A$ and $A_{i_k}$ is in the right half of $A$. Let 
$C_\ell$ be the left half and let $C_r$ be the right half of such a string $C$.
Then
\[
F_k(C, \varepsilon) = 
\ \max_{ 1 \leq j \leq k-1} \ F_j(C_\ell , \varepsilon[1 \colon j] ) + F_{k-j} (C_r, \varepsilon[j+1 \colon k] ) 
\enspace .
\]

Thus $F_k(C, \varepsilon)$ can be expressed in terms of $F_m$, for $m< k$, and its query complexity is at most $2\sum_{m=1}^{k-1} T_{|C|/2, m}$.   This is exactly the same relationship we have derived in \cref{thm:k-is} for $F_k(A, \gamma)$, and for the rest the two proofs are identical.
\end{proof}

\section{Disjunctive and minimizing problems}
\label{sec:disj}
We now turn to applications of \cref{cor:seo} and \cref{cor:meo} to directly design and analyze quantum divide and conquer algorithms. 
It turns out that we can do that for problems whose complete step is simple, and whose combine step is either the ${\rm OR}$ function or minimization. 
These problems are easily amenable to recursively applying Grover search or minimization on subproblems and therefore to quantum divide and conquer algorithms. However, let us emphasize that their create step can be rather complex and, indeed, our main examples will show such behaviour.

At the highest level, these problems will be characterized by two functions, $h: \N \rightarrow \N$ and $m:\N\rightarrow\R_+$.
The intuition behind these functions is that in the divide and conquer algorithm an instance of length $n$ will be divided into $h(n)$ instances of size $\lceil n/m(n) \rceil < n$. We call an integer $n$ {\em small}
for $m(n)$ if $\lceil n/m(n) \rceil \geq n,$ otherwise $n$ is {\em large} for $m(n)$. We will often just say that an integer is small or large, without reference to $m(n)$ when it is obvious from the context.
A string $a \in \Sigma^*$ is called {\em small} if $|a|$ is small.  These strings correspond to the base cases of the problem, and the other strings are called {\em large}.

We will consider two different classes of disjunctive and minimizing problems depending on how the instances of the subproblems are presented during the recursive calls. In the class of {\em constructible instance} problems, the instances will be computed and written explicitly to the quantum memory. In the class of {\em $t$-decomposable instance} problems, all recursive calls will be made on subsequences of the original input string and therefore can be described potentially more succinctly by a sufficient number of memory indices.

\subsection{Constructible instance problems} 
\label{subsec:explicit}
In constructible instance problems, the strings on which the recursive calls are made have to be computed explicitly, and therefore the definition of the problem on an instance doesn't make reference to the initial input. On the other hand, the definition involves the length of the initial input which determines the size of the quantum memory and the cost of all $\mathsf{QRAM}$ or $\mathsf{QRAG}$ operations made during the algorithm.

Let $h: \N \rightarrow \N$ and $m:\N\rightarrow\R_+$ be two functions. 
We say that $P: \bigcup_{n^*} (\{n^*\} \times \Sigma^{\leq n^*}) \rightarrow \{0,1\}$ (respectively $P: \bigcup_{n^*} (\{n^*\} \times \Sigma^{\leq n^*}) \rightarrow \Z$) is a {\em constructible instance $(h,m)$-disjunctive} (respectively $(h,m)$-{\em minimizing}) problem if the following two conditions are satisfied: 
\begin{enumerate}
\item
There are only a constant number of small integers for $m(n)$. 
\item
There exists a function
$$\gamma :  \bigcup_{0 < n \leq n^*} (\{n^*\} \times \Sigma^{n}  \times [h(n)] \times  \{0,1\} ) \rightarrow  \{0,1\} $$
(respectively
$$\gamma : \bigcup_{0 < n \leq n^*} (\{n^*\} \times \Sigma^{n}  \times [h(n)] \times  \Z ) \rightarrow  \Z )$$
such that for every $n^*$, for every large $n \leq n^*$, for every $a \in \Sigma^{n}$, there exist strings $a^{(1)}, \ldots , a^{(h(n))}$, with $|a^{(1)}| = \cdots = |a^{(h(n))}| = \lceil n/m(n)\rceil$, 
such that
$$P(n^*, a) = {\rm OR}_{j \in [h(n)]} \gamma(n^*, a, j, P(n^*, a^{(j)}))$$ 
(respectively $P(n^*, a) = {\rm \min}_{j \in [h(n)]} \gamma(n^*, a, j, P(n^*, a^{(j)}))$).
\end{enumerate}
For $j \in [h(n)]$, we call $a^{(j)}$ the $j$th  {\em constitutive} string of $a$ and we call the function $\gamma$ the  {\em completion} function.
A {\em constructible instance $(h,m)$-maximizing} problem is defined analogously. For every $n^*$, for a disjunctive problem $P$,
we define $P_{n^*} : \Sigma^{\leq n^*} \rightarrow \{0,1\}$ by  $P_{n^*}(a) = P(n^*, a)$, and similarly for minimizing problems.

The above conditions impose constraints on the create, complete and combine step of problems amenable to divide and conquer algorithms. The following theorem, which is valid in the $\mathsf{QRAG}$ model, stipulates the existence of a quantum algorithm for such problems.
The theorem is only stated for time complexity since for most relevant parameters the query complexity would be super-linear, and therefore trivial.

\begin{theorem}
\label{thm:recgen}
Let $P$ be a constructible instance $(h,m)$-{disjunctive} (respectively $(h,m)$-{minimizing}) problem. 
Suppose that $h(n)$ and $m(n)$ can be computed by classical algorithms in time $O(V(n))$. 
Let $N(n)$ be defined by the recursive equations $N(n) = O(1)$ when $n$ is small and, when $n$ is large, $N(n) = 2^k$, where $k$ is the smallest integer such that $2^k > \left( h(n)+1\right)\max\{n, N(n)\}$.

Let us fix $n^*$ as the input size. Suppose that there is a classical algorithm that for every large $n \leq n^*$, for every $a \in \Sigma^{n}$, computes the $h(n)$ constitutive strings $a^{(1)}, \ldots , a^{(h(n))}$ in time $O(S(n))$. 
Finally suppose that the computation of $\gamma(n^*, a, j, P(n^*, a^{(j)}))$ can be done by a quantum algorithm in time $O(G(n) \cdot \qw_{N(n^*)})$ and with error at most $1/10$.

Then there exists a quantum divide and conquer algorithm ${\cal A}$ which on an input size $n^*$ and an instance $a \in \Sigma^{n}$, written at the beginning of the memory, where $n \leq n^*$, computes $P_{n^*}$ with probability at least $9/10$
and uses a quantum memory of size $N(n^*)$.
Let $\overline{T}(n^*, n)$ denote the time complexity of the algorithm, maximized over all words in $\Sigma^n$ with $n \leq n^*$.
Then 
\[
\overline{T}(n^*, n) = O(\log  N(n^*) + \qw_{N(n^*)})  \enspace,
\]
when $n$ is small, and when $n$ is large, $\overline{T}(n^*, n)$ satisfies the recurrence equation
\[
\overline{T}(n^*,n) \leq c \sqrt{h(n)} \left(\overline{T}(n^*, n/m(n)) + G(n) \cdot \qw_{N(n^*)} + \log h(n) \right) + 
{O} \left( S(n) \cdot \qw_{N(n^*)} + V(n) \right),
\]
for some constant $c > 0$.
\end{theorem}

We remark that in \cref{thm:recgen} the size of the quantum memory is not optimized. Indeed, it is possible to use a memory of size $N'(n^*)$, where $N'(n^*)$ is defined by the recursive equations $N'(n) = O(1)$ when $n$ is small and $N'(n) = h(n) N'(n/m(n)) + n$ when $n$ is large. However, with the smaller memory, the handling of the recursive call becomes more
complicated. 
Also, for $h$ and $m$ constants, both $N(n^*)$ and $N'(n^*)$ are polynomial functions of $n$.
As mentioned in \cref{sec:instances}, in that case, when one additionally has $r = O(\log n^*)$, the cost of the memory access operations is polylogarithmic in $n^*$ for both memory sizes $N(n^*)$ and $N'(n^*)$.

Before proving \cref{thm:recgen} it is worth pointing out and illustrating a powerful corollary on achieving an almost quadratic time complexity quantum speed-up over classical divide and conquer algorithms, for the case where the completion function is trivial.

\begin{corollary}
\label{cor:classical}
Let $h$ and $m$ be constants and let $P$ be a constructible instance $(h,m)$-{disjunctive} or $(h,m)$-{minimizing} problem. 
Suppose that there is a classical algorithm that, for every large $a$ of length $n$, computes the $h$ constitutive strings $a^{(1)}, \ldots , a^{(h)}$ in time $O(S(n))$.
Finally suppose that the completion function is trivial, meaning that it is the projection to its fourth argument.

Then, there exists a quantum divide and conquer algorithm ${\cal A}$ which for every $n^*$, computes $P_{n^*}$ with probability at least $9/10$ and uses a quantum memory of size $N(n^*)$.
Let $\overline{T}(n^*, n)$ be the time complexity of the algorithm. Then
\[
\overline{T}(n^*, n) = O(\log  N(n^*) + \qw_{N(n^*)})  \enspace,
\]
when $n$ is small and when $n$ is large, $\overline{T}(n^*, n)$ satisfies the recurrence equation
\[
\overline{T}(n^*, n) \leq c \sqrt{h} \overline{T}(n^*, n/m)   + {O} \left( S(n) \cdot \qw_{N(n^*)} \right) \enspace,
\]
for some constant $c > 0$.
\end{corollary}

We illustrate \cref{cor:classical} with an example we call recursive max pooling, inspired by recent work on the power of recursion in neural networks~\cite{tahmasebi2023power}.

\begin{problem}[\RMP] 
Given $h, p\in\mathbb{N}$, an array $A$ of $n^*$ integers (i.e., $A\in\Sigma^{n^*}$ for $\Sigma=\{0,1\}^r$) where $n^*$ is a power of $p$ and, for each $j\in\{1,\ldots, \log_p n^*\}$, a set of functions $\mathcal{F}_j =\{f^{(j)}_1, \ldots, f^{(j)}_h:\Sigma^{n^*/p^{j-1}}\rightarrow\Sigma^{n^*/p^j}\}$, compute $P_{n^*}(A)$ where $P_{n^*}:\Sigma^{\le n^*}\rightarrow\mathbb{Z}$ is defined recursively:
$$
P_{n^*}(B) = \begin{cases} 
B &\text{if }|B| = 1\\ 
\max \{{P_{n^*}(B^{(1)}), \ldots, P_{n^*}(B^{(h)})}\} & \text{otherwise}
\end{cases}
$$
where $B^{(i)} = f_i^{(j)}(B)$ when $|B| = n^*/p^{j-1}$.
\end{problem} 

One can view the strings $B^{(i)}$ as being computed, for example, by $h$ neural networks, each downsizing their common input by a factor of $p$, and then max pooling applied recursively to determine a final network output value. Note that our model differs from the recursive neighbourhood pooling graph neural network of~\cite{tahmasebi2023power}: in that work, the subproblems that correspond to our $B^{(i)}$ are subgraphs of a parent graph, and the pooling function they use is required to be injective. On the other hand, while maximum pooling as used here is highly lossy, we allow greater flexibility in constructing the substrings. Our problem is also similar in spirit to (a recursive version of) convolutional neural networks, which use multiple convolutional filters (which can be viewed as special cases of our $f^{(j)}_i$) to produce a number of downsized images which are then pooled and further processed. 

\begin{theorem}
    Let $h = p^k$ for $k\in\mathbb{N}, k\ge 4$. If the $f^{(j)}_i(B)$ can be computed classically in time $O(|B|^2)$ for any input $B$, \RMP\ can be solved by a quantum algorithm with running time $\overline{T'}(n) = O((n)^{k/2 + \log_hc}\cdot \qr_{N(n)})$.
\end{theorem}

\begin{proof}
    \RMP\ is a constructible instance $(h,m)$-minimizing problem with $h=p^k$, $m=p$, and trivial completion function, where the constituent strings are computed by the functions $f_i^{(j)}$. From \cref{cor:classical}, the running time recursion satisfies $\overline{T}(n^*, n) = ch^{k/2} \overline{T}(n^*, n/h) + O(n^2 \cdot \qr_{N(n^*)})$. Taking $\overline{T'}(n) = \overline{T}(n, n)$, the result follows. 
\end{proof}
The complexity $T$ of the obvious classical divide and conquer algorithm (ignoring memory access costs) for $P$ satisfies $T(n) = p^k T(n/p) + O(n^2)$ resulting in $T(n) = O(n^k)$. 
If $\qr_{N(n^*)}$ is polylogarithmic in  $n^*$, the quantum result achieves an almost quadratic improvement compared with classical when $h$ is large compared to $c$. 

We now prove \cref{thm:recgen}.
\begin{proof}
For every $u \in \{0,1\}^*$, with $|u| < \log N(n^*)$, we denote by $M_u$ the section of the memory content register indexed by binary strings with prefix $u$. For example $M_{\epsilon}$ is the full memory, for the empty string $\epsilon$.
In the run of ${\cal A}$, an instance $a$, with $|a|=n$, will be written at the beginning of $M_u$, a section of memory of size $N(n)$, where $u$ is chosen recursively as follows. The input of length $n^*$ is written at the beginning of $M_{\epsilon}$.  If an instance $a$ is written at the beginning of $M_u$, its $j$-th constitutive string $a^{(j)}$ will be written at the beginning of $M_{uv_j}$, where $v_j$ is the integer $j$ written in binary, and $uv_j$ denotes the concatenation of $u$ and $v_j$. When dealing with $a$, the algorithm will have access to $u$.
The function $N(n)$ is large enough that all instances in the subsequent recursive calls can be written to $M_u$.

A small instance, when its location in the memory is known,  can be accessed in time $O(\qw_{N(n^*)}),$ and therefore the problem can be solved in time $O(\log N(n^*) + \qw_{N(n^*)}$, accounting for the time required to input the address.

On a large instance $a$, written at the beginning of $M_u$, in the create step ${\cal A}$ computes sequentially, for every $j \in [h(n)]$, the constitutive string $a^{(j)}$. Then, for $j \in [h(n)]$, it writes $a^{(j)}$ at the beginning of $M_{uv_j}$. 
To compute $P_{n^*}(a)$, then it uses \cref{cor:seo} (respectively \cref{cor:meo}) with the $h(n)$ functions $f_1, \ldots , f_{h(n)}$, where $f_j(a) = \gamma(n^*, a, j, P_{n^*}(a^{(j)}))$.

For this, we describe a quantum algorithm $F$ that on input $\ket{j}\ket{a}$ computes $f_j(a)$. 
By the recursive hypothesis, we have an algorithm ${\cal B}$ that computes $P_{n^*}$ on instances of size $n/m(n)$, when the instance is written at the beginning of $M_{uv_j}$.
On $\ket{j}\ket{{a}}$ the algorithm $F$ first copies $uv_j$  at the beginning of the memory index register and then it executes ${\cal B}$.  
Once the computation of $P_{n^*}(a^{(j)})$ is finished, it applies the completion function on $(n^*, a, j, P_{n^*}(a^{(j)})$. This ends the description of $F$.
By applying \cref{cor:seo} or \cref{cor:meo}, ${\cal A}$ finds the value of $P_{n^*}$ on $a$. 

We show that the error of algorithm ${\cal A}$ is at most $9/10$ which comes from the application of \cref{cor:seo} (respectively \cref{cor:meo}).
For this, we claim that the error for computing $F$ is at most $2/10$.
Indeed the error of the recursive call ${\cal B}$ is at most 1/10, and the completion function can be computed, by hypothesis, with error at most $1/10$.

The create step costs ${O} \left( S(n)\cdot\qw_{N(n^*)} + V(n) \right)$ which includes also the writing of the constitutive strings into memory.
Completing each recursive call costs  $O(G(n) \cdot \qw_{N(n^*)})$ and therefore the time of implementing $F$ is $O(\overline{T}(n^*, n/m(n)) + G(n) \cdot \qw_{N(n^*)})$. 
Therefore \cref{cor:seo} (respectively \cref{cor:meo}) implies the recurrence equation for the time complexity.
\end{proof}

\subsection{$t$-decomposable instance problems} 
\label{subsec:self-reducible}

In a constructible instance problem, the main difficulty can be in the create step, computing the constitutive strings which may be non-trivially related to the original input string. The situation is very different when, for every instance $a$, the constitutive strings of $a$ are subsequences of the original input string $\alpha$. In particular, we will consider \emph{t-decomposable instance} problems where, at every level of recursion, every constitutive string is the concatenation of $t$ substrings of the input, for some constant $t$.

For some integers ${n^*} \geq n > 0,$ let $\alpha \in \Sigma^{n^*}$ and let $t > 0$ be an integer constant.
We say that $\mathcal{I} \in [n^*]^{2t}$ is an $n^*$-{\em valid} $t$-{\em description}  if $\mathcal{I} = (b_1, e_1, \ldots, b_t, e_t)$ with $1\leq b_1 \leq e_1 < b_2 \leq  \ldots  < b_t \leq e_t \leq {n^*}$.
We define the {\em size} of $\mathcal{I}$ as $s(\mathcal{I}) = \sum_{i=1}^t (e_i - b_i +1)$.
Let $\mathcal{V}(n^*,t)$ denote the set of $n^*$-{valid} $t$-descriptions, and for $0 < n \leq n^*$, let $\mathcal{V}(n^*,t,n)$ denote the set of $n^*$-{valid} $t$-descriptions of size $n$.
For $\mathcal{I} \in  \mathcal{V}(n^*,t)$, where $\mathcal{I} = (b_1, e_1, \ldots, b_t, e_t),$ we denote by $\alpha_{\mathcal{I}}$ the string $\alpha[b_1:e_1]\concat \cdots \concat \alpha[b_t:e_t]$.
Observe that $|\alpha_{\mathcal{I}}| = s(\mathcal{I}).$
Let $a \in \Sigma^n,$ we say that $a$ is {\em $t$-decomposable} in $\alpha$ if there exists $\mathcal{I} \in  \mathcal{V}(n^*,t, n)$ such that $a = \alpha_{\mathcal{I}},$ and we call $\mathcal{I}$ a {\em $t$-description} of $a$ in $\alpha$. 
We are interested in $t$-decomposable instance problems, meaning that for every $\alpha \in \Sigma^{n^*}$, for every large $t$-decomposable subsequence $a$ of $\alpha$, the constitutive strings of $a$ are all $t$-decomposable in $\alpha$.

Note that is important that the constitutive strings of an instance $a$ are not only $t$-decomposable in $a$ but also $t$-decomposable in $\alpha$. To see why, consider an instance $a$ that is $2$-decomposable in $\alpha$. If its constitutive strings are $2$-decompasable in $a$, then we can only claim for them $3$-decomposability in $\alpha$, and so on. Therefore, even starting with a simple input $\alpha$, we obtain more and more complex subsequences of $\alpha$ at the deeper layers of the recursion that can be difficult to deal with.

We now formally define $t$-decomposable instance problems.  
Let $\alpha \in \Sigma^{n^*}$ be an input. 
An instance $a \in \Sigma^n$ which is a $t$-decomposable in $\alpha$ will be specified by an $n^*$-valid $t$-description 
$\mathcal{I}$ of $a$ in $\alpha$ of size $n$. The $t$-description can be given in some work registers which is distinct from the memory registers.
Observe that $\mathcal{I}$ is an $O(\log {n^*})$ string.
We will involve in the definition the create functions $\delta$ where, for every $\alpha \in \Sigma^{n^*}$, for every $\mathcal{I} \in  \mathcal{V}(n^*,t,n)$, and for every $j \in [h(n)]$, the value of $\delta(\alpha, \mathcal{I}, j)$ is a $t$-description of the $j$th constitutive string $\alpha_{\mathcal{I}}^{(j)}$ of $\alpha_{\mathcal{I}}.$

Let $h: \N \rightarrow \N$ and $m:\N\rightarrow\R_+$ be two functions and $t > 0$ an integer constant.
We say that $P: \bigcup_{n^*} (\Sigma^{n^*} \times \mathcal{V}(n^*,t)) \rightarrow \{0,1\}$ (respectively $P: \bigcup_{n^*} (\Sigma^{n^*} \times \mathcal{V}(n^*,t))  \rightarrow \Z$) is a $t$-{\em decomposable instance} $(h,m)$-{\em disjunctive} (respectively $(h,m)$-{\em minimizing}) problem if the following two conditions are satisfied:
\begin{enumerate}
\item
There are only a constant number of small integers for $m(n)$.
\item
There exist two functions
$$\delta : \bigcup_{0 < n \leq n^*} (\Sigma^{n^*} \times \mathcal{V}(n^*,t, n) \times [h(n)]) \rightarrow 
\mathcal{V}(n^*,t,  \lceil n/m(n) \rceil)$$
and
$$\gamma : \bigcup_{0 < n \leq n^*} (\Sigma^{n^*} \times \mathcal{V}(n^*,t, n) \times [h(n)] \times 
\mathcal{V}(n^*,t,  \lceil n/m(n) \rceil \times  \{0,1\} )) \rightarrow  \{0,1\} $$
(respectively
$$\gamma : \bigcup_{0 < n \leq n^*} (\Sigma^{n^*} \times \mathcal{V}(n^*,t, n) \times [h(n)] \times 
\mathcal{V}(n^*,t,  \lceil n/m(n) \rceil \times  \Z )) \rightarrow  \Z )$$
such that for every $\alpha \in \Sigma^{n^*}$, for every 
$\mathcal{I} \in  \mathcal{V}(n^*,t,n)$ where $n$ is large,
$$P(\alpha, \mathcal{I}) = {\rm OR}_{j \in [h(n)]} \gamma(\alpha, \mathcal{I}, j, \delta(\alpha, \mathcal{I}, j), 
P(\alpha, \delta(\alpha, \mathcal{I}, j)))$$
(respectively 
$P(\alpha, \mathcal{I}) = {\rm \min}_{j \in [h(n)]} \gamma(\alpha, \mathcal{I}, j, \delta(\alpha, \mathcal{I}, j), 
P(\alpha, \delta(\alpha, \mathcal{I}, j)))$).
\end{enumerate}
We call the functions $\delta$ and $\gamma$  respectively the {\em create} and the {\em completion} function, and for $j \in [h(n)]$, we call $\alpha_{\delta(\alpha, \mathcal{I}, j)}$ the $j$th  {\em constitutive} string of $\alpha_{\mathcal{I}}$.
A {\em $t$-decomposable instance $(h,m)$-maximizing} problem is defined analogously. 

Our next theorem stipulates the existence of a quantum divide and conquer algorithm for $t$-decomposable instance disjunctive or minimizing problems. 
The complexity of the algorithm is expressed as a function of the complexities of the create and completion functions.
The result is valid in both the $\mathsf{QRAM}$  and the $\mathsf{QRAG}$ model.

\begin{theorem}
\label{thm:rec}
For some constant $t$, let $P$ be a $t$-decomposable instance $(h,m)$-{disjunctive} (respectively $(h,m)$-{minimizing}) problem.
Suppose that $h(n)$ and $m(n)$ can be computed by classical algorithms in time $O(V(n))$ and without queries.

Suppose that there is a quantum algorithm that for every input string $\alpha \in \Sigma^{n^*}$ given in the quantum memory,
for every large $n$, for every $\mathcal{I} \in  \mathcal{V}(n^*,t,n)$, computes $ \delta(\alpha, \mathcal{I}, 1), \ldots,  \delta(\alpha, \mathcal{I}, h(n))$ with $O(d(n))$ queries, in time $O(D(n) \cdot \qr_{n^*})$ in the $\mathsf{QRAM}$ model and 
in time $O(D(n) \cdot \qw_{M(n^*)})$ in the $\mathsf{QRAG}$ model, for some function $M$.
Suppose also that each computation has probability of error at most $1/10$.

Let $N'(n)$ be defined by the recursive equations $N'(n) = O(1)$ when $n$ is small and $N'(n) = N'(n/m(n)) + M(n)$
when $n$ is large.
Also define $N(n^*) = n^* + N'(n^*)$. Finally suppose that for every input string $\alpha \in \Sigma^{n^*}$, for every large $n$, for every $\mathcal{I} \in  \mathcal{V}(n^*,t,n)$, the computation of the completion function $\gamma$ can be done 
by a quantum algorithm with $O(g(n))$ queries, in time $O(G(n) \cdot \qw_{N(n^*)} )$ and with error at most $1/10$.

Then there exists a quantum divide and conquer algorithm ${\cal A}$ which computes the problem $P$ on $t$-decomposable instances in $\alpha$ with probability at least $8/10$ and in $\mathsf{QRAG}$ model uses a quantum memory of size $N(n^*)$. Let denote by
$T(n^*,n)$ (respectively $\overline{T}(n^*,n)$) its query (respectively time) complexity, maximized over all inputs $\alpha$ of length $n^*$ and over all $t$-decomposable instances in $\alpha$ of length $n$, specified by a $t$-description in $\alpha$. Then, for small $n$, we have
\[
T(n^*,n) = O(1)  \enspace
\]
in both memory models, and
\[
  \overline{T}(n^*,n) =\begin{cases}
  O(\log n^* + \qr_{n^*}) &  (\mathsf{QRAM} \text{ model})\\
  O(\log N(n^*) + \qw_{N(n^*)}) &(\mathsf{QRAG} \text{ model}).
  \end{cases}
\]
When $n$ is large, $T(n^*,n)$ and $\overline{T}(n^*,n)$ satisfy, for some constant $c > 0$, the recurrences
\[
T(n^*,n) \leq c \sqrt{h(n)} \left(T(n^*,n/m(n)) + g(n) \right)  + {O}(d(n))  \enspace,
\]
in both memory models, and
\[
\resizebox{\textwidth}{!}{
$
\overline{T}(n^*, n) \le \begin{cases}
c \sqrt{h(n)} \left(\overline{T}(n^*,n/m(n))  + G(n) \cdot \qw_{n^*} + \log h(n) \right) + 
{O} \left(D(n)  \cdot \qr_{n^*} + V(n) \right) & (\mathsf{QRAM}) \\
c \sqrt{h(n)} \left(\overline{T}(n^*, n/m(n))  + G(n) \cdot \qw_{N(n^*)} + \log h(n) \right) + 
{O} \left(D(n)  \cdot \qw_{N(n^*)} + V(n) \right) & (\mathsf{QRAG}).
\end{cases}
$
}
\]
\end{theorem}

\begin{proof}
A small instance can be solved with a constant number of queries, therefore we have $T(n^*,n)= O(1)$ both in the $\mathsf{QRAM}$ and the $\mathsf{QRAG}$ model when $n$ is small. Given a quantum memory of content size $s$, copying an index to the memory index register takes time $\log s$. Taking $s=n^*$ in the $\mathsf{QRAM}$ model and $s = N(n^*)$  in the $\mathsf{QRAG}$ model, we get the respective equalities for $\overline T(n^*,n)$, when $n$ is small. 

Let us now suppose that we are given $\mathcal{I} \in  \mathcal{V}(n^*,t,n)$.
The algorithm ${\cal A}$ first computes $h(n)$ and sequentially $\delta(\alpha, \mathcal{I}, 1), \ldots,  \delta(\alpha, \mathcal{I}, h(n))$.
Then, to compute $P(\alpha, \mathcal{I}) $, the algorithm ${\cal A}$ uses in the case of a search problem \cref{fact:qseo} for the query complexity and \cref{cor:seo} for the time complexity (respectively, for minimizing problems, \cref{fact:qmeo} and \cref{cor:meo}) with the functions $f_j(\mathcal{I}) =  \gamma(\alpha, \mathcal{I}, j, \delta(\alpha, \mathcal{I}, j), 
P(\alpha, \delta(\alpha, \mathcal{I}, j)))$, for $j \in [h(n)]$.
For this we describe a quantum algorithm $F$ that on input $\ket{j}\ket{\mathcal{I}}$ computes $f_j(\mathcal{I})$.
By the recursive hypothesis, we have an algorithm that computes $P$ on $t$-decomposable strings of length $n/m(n)$, specified by $t$-description. $F$ runs this algorithm on  $\alpha_{\delta(\alpha, \mathcal{I}, j)}$ which was already computed, and then applies the completion function to the result.
Therefore ${\cal A}$ computes the value of $P(\alpha, \mathcal{I}) $.

The error of algorithm ${\cal A}$ comes from the computation of the $t$-descriptions of the constitutive strings and from the application of the respective \cref{fact:qseo} or \cref{cor:seo} (or \cref{fact:qmeo} or \cref{cor:meo} for minimizing problems).
The error for computing a $t$-description is, by hypothesis, at most $1/10$. 
Let us consider the error for computing $F$. When the description of the $j$th constitutive string is correctly computed then 
the error of the recursive call is at most $2/10$. The completion function can be computed, by hypothesis, with error at most $1/10$.
Thus the error of $F$ is at most $3/10$. 
Therefore the error of ${\cal A}$ coming from \cref{fact:qseo} or \cref{cor:seo} (respectively \cref{fact:qmeo} or \cref{cor:meo} for minimizing problems) is also at most $1/10$, and its overall error is at most $2/10$.

Computing $h(n)$ and $m(n)$ takes no queries and time $O(V(n))$, and computing the $t$-descriptions of all constitutive strings
takes time $O( D(n) \cdot \qr_{n^*})$ and $O(d(n))$ queries. 
Computing the completion function $\gamma$ takes $O(g(n))$ queries,  time $O(G(n) \cdot \qr_{n^*}  )$ in the $\mathsf{QRAM}$ model and time $O(G(n) \cdot \qw_{N(n^*)}  )$ in the $\mathsf{QRAG}$ model. Thus the query cost of implementing $F$ is $T(n^*,n/m(n)) +g(n)$ and can be done in time $O(\overline{T}(n^*,n/m(n)) + G(n)  \cdot \qr_{n^*})$ in the $\mathsf{QRAM}$ model and in time $O(\overline{T}(n^*,n/m(n)) + G(n)  \cdot \qw_{N(n^*)})$ in the $\mathsf{QRAG}$ model.
Therefore \cref{fact:qseo} (respectively \cref{fact:qmeo}) implies the recurrence equation for the query complexity and \cref{cor:seo} (respectively \cref{cor:meo}) implies the recurrence equation for the time complexity.
The memory size used for computing in superposition the constitutive substrings, over all levels of the recursion, is $N'(n^*)$
therefore counting also the input length, the algorithm can be implemented in the $\mathsf{QRAG}$ model with memory content size $N(n^*)$.
\end{proof}

\section{Longest Distinct Substring}
\label{sec:lds}
Recall that a substring of $\alpha$ is a subsequence of consecutive symbols. 
We say that a substring is \emph{distinct} if no character in the substring appears more than once.  The \ldsf \ (\lds) problem is to find the length of a longest distinct substring of $\alpha$.

\begin{problem}[\ldsf]
Let $\Sigma$ be an alphabet and $\alpha \in \Sigma^{n^*}$ be a string. The goal is to output the length of a longest distinct substring of $\alpha$, denoted $\lds(\alpha)$.
\end{problem}

Computing $\lds(\alpha)$ is naturally a maximization problem over all substrings of $\alpha$ and so fits into the bottom-up divide and conquer framework of \cref{sec:bottomup}.  By repeating the last character of $\alpha$ we may assume that its size is a power of 2 without affecting $\lds(\alpha)$.  
Letting $P(\alpha, t, k)$ be the length of a longest distinct substring of $\alpha[(k-1)n/2^t+1 \upto k n/2^t]$ whose left endpoint is in the left half of this interval and the right endpoint is in the right half of the interval, we see that $\lds(\alpha) = \max_t \max_{1 \le k \le 2^t} P(\alpha, t, k)$.
Our main task is thus to compute $P(\alpha, t, k)$. This will again be done with a divide and conquer algorithm, but now one with a non-trivial create step.  
Because of the nature of the recursive calls in the divide and conquer algorithm for $P(\alpha,t,k)$, in the next subsection we develop an algorithm for a more general version of the problem which we call \bldsf\ (\blds).

\subsection{Bipartite longest distinct substring}
\begin{problem}[\bldsf] 
$\alpha \in \Sigma^{n^*}$ be a string. For a $n^*$-valid 2-description $\Iset$ let $\Iset(k)$ refer to the $\kth$ element of $\Iset$.  
Given $\alpha$ and a $n^*$-valid 2-description $\Iset$, the goal of the bipartite longest distinct substring problem is to output the length of a longest distinct substring of $\alpha[\Iset(1) \colon \Iset(2)] \concat \alpha[\Iset(3) \colon \Iset(4)]$ that includes at least one of $\alpha_{\Iset(2)}, \alpha_{\Iset(3)}$. 
We denote this value by $\blds(\alpha, \Iset)$. 
\end{problem}

A key role in the algorithm for $\blds$ is played by finding the longest distinct substring whose right or left endpoint is given by an index $k$. We make a couple of definitions related to this problem.

\begin{definition}[Longest distinct substring with constrained endpoint]
Let $\Sigma$ be an alphabet and $\alpha \in \Sigma^{n^*}$ be a string.  For $X \subseteq [n^*]$, let $\lds(\alpha,X)$ be the length of a longest distinct substring of $\alpha$ whose right endpoint is in $X$.
\end{definition}

\begin{definition}[Longest distinct substring with a fixed endpoint]
Let $\Sigma$ be an alphabet, $\alpha \in \Sigma^{n^*}$ be a string, and $\Iset$ an $n^*$-valid 2-description.  
Let $k \in \{\Iset(1), \ldots, \Iset(2)\} \cup \{\Iset(3), \ldots, \Iset(4)\}$. 
Define $L(\alpha, \Iset, k)$  to be the left endpoint of the longest distinct substring of $\alpha[\Iset(1) \upto \Iset(2)] \concat \alpha[\Iset(3) \upto \Iset(4)]$ whose right endpoint is $k$.
Analogously, define $R(\alpha,\Iset, k)$ to be the right endpoint of the longest distinct substring of $\alpha[\Iset(1) \upto \Iset(2)] \concat \alpha[\Iset(3) \upto \Iset(4)]$ whose left endpoint is $k$.  
\end{definition}

The high-level idea of the algorithm is the following.  We first compute the smallest index $i_1 \ge \Iset(1)$ such that $\alpha[i_1 \upto \Iset(2)]$ is a distinct substring and the largest index $j_2 \le \Iset(4)$ such that $\alpha[\Iset(3) \upto j_2]$ is a distinct substring.
Clearly, $\blds(\alpha, \Iset) = \blds(\alpha, \Jset)$ for $\Jset = (i_1, \Iset(2), \Iset(3), i_2)$.

$\blds(\alpha, \Jset)$ will be the maximum of $n_1 = \Jset(2) - \Jset(1) + 1, n_2 = \Jset(4) - \Jset(3) + 1$ and the length of a longest distinct substring whose left endpoint is in $\{\Jset(1), \ldots, \Jset(2)\}$ and right endpoint is in $\{\Jset(3), \ldots, \Jset(4)\}$.  Call the latter the \emph{crossing value}.

The algorithm uses a divide and conquer approach to computing the crossing value. We may assume without loss of generality that $n_2 \le n_1$, as the length of a longest distinct substring of a string and its reversal are the same.  For a constant $h$, we partition $\{\Jset(3), \ldots, \Jset(4)\}$ into $h$ intervals.  The right endpoint of a substring realizing the crossing value must lie in one of these intervals.
Thus the main subproblem of our divide and conquer approach is to compute the length of a longest distinct crossing substring whose right endpoint is contained in an interval.  The core idea of how to reduce this problem to a smaller instance of \blds\ is contained in the following lemma.

\begin{lemma}
\label{lem:lds}
Let $\Sigma$ be an alphabet and $a = a_1, \ldots, a_N \in \Sigma^N$. 
Suppose that $1 \leq u < v \leq N$ are such that $a[1\upto v]$ is a distinct substring and $a[u \upto N]$ is a distinct substring.  
Let $b = a[1\upto u-1] \concat a[v+1\upto N]$.  Then
\[
\lds(a) = v - u + 1 + \lds(b) \enspace .
\]
\end{lemma}

\begin{proof}
First we show that $\lds(a) \ge v - u + 1 + \lds(b)$.  
Let $\hat b$ be a distinct substring of $b$ realizing $\lds(b)$ and let $\hat b_\ell$ be the substring of $\hat b$ to the left of $a_{u}$ and $\hat b_r$ the substring of $\hat b$ to the right of $a_{v}$.  Then $\hat b_\ell \concat a[u \upto v] \concat \hat b_r$ is a distinct substring of $a$.  The portion $\hat b_\ell \concat a[u \upto v]$ is distinct because $a[1 \upto v]$ is, the portion $a[u \upto v] \concat \hat b_r$ is distinct because $a[u \upto N]$ is, and there is no collision between $\hat b_\ell$ and $\hat b_r$ because $\hat b$ is distinct.

Now we show that $\lds(a) \le v - u + 1 +\lds(b)$.
Let $\hat a$ be a substring of $a$ realizing $\lds(a)$ and split up $\hat a$ as $\hat a_\ell, \hat a_m, \hat a_r$ for the portion strictly to the left of $u$, between $u$ and $v$, and strictly to the right of $v$, respectively.  If any of 
$\hat a_\ell, \hat a_m, \hat a_r$ are empty, then clearly $\lds(a) \le v - u + 1 +\lds(b)$.  
If they are all nonempty, then $\hat a_\ell, \hat a_r$ must form a distinct substring of $b$, thus also $\lds(a) \le v - u + 1 +\lds(b)$.
\end{proof}

The next proposition gives the application of \cref{lem:lds} to finding the longest distinct substring whose right endpoint is in a given interval.
To state the proposition the following definition will be useful.
\begin{definition}
Let $N$ be a positive integer and $i < j \in [N]$.
For $k \ge j$, let $\prev_{i, j}(k) = i$ if $k = j$ and $\prev_{i, j}(k) = k-1$ otherwise.  For $k \le i$, let $\suc_{i, j}(k) = j$ if $k = i$ and $\suc_{i, j}(k) = k+1$ otherwise.
\end{definition}

\begin{proposition}
\label{prop:blds}
Let $\alpha \in \Sigma^{n^*}$ and $\Iset$ be a $n^*$-valid 2-description such that $\alpha[\Iset(2)] \concat \alpha[\Iset(3) \upto \Iset(4)]$ is distinct.
For $\Iset(3) \le k_1 \le k_2 \le \Iset(4)$ let $\Kset = \{k_1, \ldots, k_2\}$, $t = |\Kset|$ and $\Kset^+ = \prev_{\Iset(2), \Iset(3)}(k_1) \cup \Kset$. Further, define
\begin{align*}
i_1' &= L(\alpha, \Iset, \prev_{\Iset(2), \Iset(3)}(k_1)), \\
i_2' &= \min(\Iset(2), i_1' + t - 1), \\
k_2' &= \min(k_2, R(\alpha, \Iset, \suc_{\Iset(2), \Iset(3)}(i_2'))) \enspace .
\end{align*}
Let $b = \alpha[\Iset(1) \upto \Iset(2)] \concat \alpha[\Iset(3) \upto \Iset(4)]$ and $\Jset = (i_1', i_2', k_1, k_2')$.
Then
\[
\lds(b, \Kset^+)
=  \Iset(2) - \Jset(2) + \Jset(3) - \Iset(3) + \blds(\alpha, \Jset) \enspace.
\]
\end{proposition} 

\begin{proof}
Let $\Lset = \{k_1, \ldots, k_2'\}$ and $\Lset^+ = \{\prev_{\Iset(2), \Iset(3)}(k_1)\} \cup \Lset$.  
We first claim that $\lds(b, \Kset^+) = \lds(b, \Lset^+)$.
We establish this by showing that $\lds(b, \Kset \setminus \Lset) < \lds(b, \{\prev_{\Iset(2), \Iset(3)}(k_1)\})$.
We may assume $k_2' < k_2$, as otherwise $\Kset \setminus \Lset = \emptyset$ and there is nothing to prove.  The fact that $k_2' < k_2$ has two important consequences: 
\begin{enumerate}
\item $\suc_{\Iset(2), \Iset(3)}(i_2') = i_2' + 1 < \Iset(2)$, since by assumption $\alpha[\Iset(2)] \concat \alpha[\Iset(3) \upto \Iset(4)]$ is distinct.
\item We have $i_2' - i_1' = t - 1$ since $i_2' < \Iset(2)$ by the previous item.
\end{enumerate}

By definition of $k_2'$ the string $\alpha[\suc_{\Iset(2), \Iset(3)}(i_2') \upto \Iset(2)] \concat \alpha[\Iset(3) \upto k_2'+1]$ is not distinct (since we are in the case $k_2' < k_2$), meaning that the left endpoint of a longest distinct string with right endpoint in $\Kset \setminus \Lset$ is at least $i_2' + 2$, using item (1) above.
Thus the length of a longest substring of $b$ whose right endpoint is in $\Kset \setminus \Lset$ is at most 
\begin{align*}
\Iset(2) - (i_2' + 2) + 1 + k_2 - \Iset(3) + 1 &= \Iset(2) - i_2' + (k_1 + t - 1) - \Iset(3) \\
&= (\Iset(2) - i_1') - (i_2' - i_1') - 1 + k_1 - \Iset(3) + t - 1 \\
&= \Iset(2) - i_1' + k_1 - \Iset(3) \enspace,
\end{align*}
because $i_2' - i_1' = t - 1$ by item (2) above. On the other hand, $\lds(b, \{\prev_{\Iset(2), \Iset(3)}(k_1)\}) =  \Iset(2) - i_1' + k_1 - \Iset(3) + 1$, thus the claim follows.

Let $D = \Iset(2) - i_2' + k_1 - \Iset(3) + \blds(\alpha, \Jset).$
We now show that $D = \lds(b, \Lset^+).$
Let $b' = \alpha[i_1' \upto \Iset(2)] \concat \alpha[\Iset(3) \upto k_2'].$
By definition of $i_1'$, we have $\lds(b, \Lset^+) = \lds(b')$.
Also by definition, 
$\alpha[i_1' \upto \Iset(2)] \concat \alpha[\Iset(3) \upto \prev_{\Iset(2), \Iset(3)}(k_1)]$ is distinct and 
$\alpha[\suc_{\Iset(2), \Iset(3)}(i_2') \upto \Iset(2)] \concat \alpha[\Iset(3) \upto k_2']$ is distinct.  
Applying \cref{lem:lds} to $b'$ says that $\lds(b') = D$, giving the proposition.
\end{proof}

\begin{corollary}
\label{cor:h}
Let $\alpha \in \Sigma^{n^*}$ and $\Iset$ a $n^*$-valid 2-description such that $\alpha[\Iset(2)] \concat \alpha[\Iset(3) \upto \Iset(4)]$ is distinct.
Let $\Jset_0, \ldots, \Jset_{h-1}$ be a partition of $\{\Iset(3), \ldots, \Iset(4)\}$ into intervals, and let $x_\ell, y_\ell$ be respectively the left and right endpoints of $\Jset_\ell$.    For $\ell = 0, \ldots, h-1$ let 
\begin{align*}
p_\ell &= L(\alpha, \Iset, \prev_{\Iset(2), \Iset(3)}(x_\ell)), \\
q_\ell &= \min(\Iset(2), p_\ell+|\Jset_\ell|-1), \\
y_\ell'&= \min(y_\ell, R(\alpha, \Iset, \suc_{\Iset(2), \Iset(3)}(q_\ell))) \enspace,
\end{align*}
and $\Iset_\ell = (p_\ell, q_\ell, x_\ell, y_\ell')$.
Then
\[
\blds(\alpha, \Iset) = \max_{0 \le \ell \le h-1} 
\Iset(2) - \Iset_\ell(2) + \Iset_\ell(3) - \Iset(3) + \blds(\alpha, \Iset_\ell) \enspace .
\]
\end{corollary}

\begin{proof}
We may assume the partition is such that $x_0 = \Iset(3)$.
Let $b = a[p_0 \upto \Iset(2)] \concat a[\Iset(3) \upto \Iset(4)]$ so that $\blds(\alpha, \Iset) = \lds(b)$.  
As $\alpha[p_0 \upto \Iset(2)]$ is distinct and the $\Jset_\ell$ partition $\{\Iset(3), \ldots, \Iset(4)\}$, we have
\[
\lds(b) = \max_{0 \le \ell \le h - 1} \lds(b, \prev_{\Iset(2), \Iset(3)}(x_\ell) \cup \Jset_\ell) \enspace.
\]
By \cref{prop:blds}, 
\[
\lds(b, \prev_{\Iset(2), \Iset(3)}(x_\ell) \cup \Jset_\ell) = \Iset(2) - 
\Iset_\ell(2) + \Iset_\ell(3) - \Iset(3) + \blds(\alpha, \Iset_\ell) \enspace .
\]
The corollary follows.
\end{proof}

\begin{algorithm}[!htbp]
\caption{$\texttt{bipartiteLDS}_h(\alpha, \Iset)$}
\label{alg:blds}
\begin{algorithmic}[1]
  \State $i_1 \gets L(\alpha, \Iset, \Iset(2)), r \gets R(\alpha, \Iset, \Iset(3)), 
    n_2' \gets r - \Iset(3) + 1$ \label{alg:line_one}
  \State $j_2 \gets R(\alpha, \Iset, \Iset(2)]$ \label{alg:line_two}
  \State $n_1 \gets \Iset(2) - i_1 + 1, n_2 \gets j_2 - \Iset(3) + 1$. By either working 
  with $\alpha[i_1 \upto \Iset(2)] \concat \alpha[\Iset(3) \upto j_2]$ or its reversal and renaming 
  indices as needed we may assume $n_2 \le n_1$.
  \If{$n_2 = 1$}
  \label{alg:base_case}
  \Comment{Base Case}
    \State $m \gets  (\alpha[\Iset(3)] \in \alpha[i_1 \colon \Iset(2)]) \; ? \; n_1 \upto n_1 + 1$
    \State return $\max(m, n_2')$. \label{alg:return1}
  \EndIf
  \State $\rstart \gets \Iset(3)$
  \State $\currentmax \gets n_2'$
  \For {$\ell = 0; \; \ell < h; \; \ell = \ell + 1$}
    \State $t \gets (\ell < n_2 \bmod h) \; ? \; \ceil{n_2/h} \upto \floor{n_2/h}$
    \State $\rend \gets \rstart + t - 1$
    \State $\lstart \gets L(\alpha, (i_1, \Iset(2), \Iset(3), j_2), \prev_{\Iset(2), \Iset(3)}(\rstart))$
    \State $\lend \gets \min(\Iset(2), \lstart + t - 1)$
    \State $\rmiddle \gets \min(\rend, R(\alpha, (i_1, \Iset(2), \Iset(3), j_2), \suc_{\Iset(2), \Iset(3)}(\lend)])$
    \State $v_\ell \gets \Iset(2) - \lend + \rstart - \Iset(3) + \texttt{bipartiteLDS}_h(\alpha, 
(\lstart, \lend, \rstart, \rmiddle))$
    \State $\currentmax \gets \max(\currentmax, v_\ell)$
    \State $\rstart \gets \rend + 1$
  \EndFor
  \State return $\currentmax$. \label{alg:return2}
\end{algorithmic}
\end{algorithm}

We now prove the main result of this subsection, an upper bound on the quantum time complexity of $\blds$.  This result is in the $\mathsf{QRAG}$ model due to the use of $\mathsf{QRAG}$ gates in Ambainis' element distinctness algorithm 
\cite{Amb07}.

\begin{fact}
\label{fact:ed_binary}
Let $\Sigma$ be an alphabet, $\alpha \in \Sigma^{n^*}$
There is a quantum algorithm that for any $n^*$-valid 2-description $\Iset$ for $\alpha$ with $s(\Iset) = n$ and
$k \in \{\Iset(1),\ldots, \Iset(2)\} \cup \{\Iset(3), \ldots, \Iset(4)\}$ outputs $R(\alpha, \Iset, k)$ with probability at least $9/10$ in time $\tilde O(n^{2/3}) \cdot \qw_{O(n^*)}$ and with $O(n^{2/3} \log(n))$ queries, given oracle access to $\alpha$.  The same holds for $L(\alpha, \Iset, k)$.
\end{fact}

\begin{proof}
Using Ambainis' element distinctness algorithm from \cref{fact:ed} we can check if any substring of $\alpha[\Iset(1) \colon \Iset(2)] \concat \alpha[\Iset(3) \colon \Iset(4)]$ is distinct with success probability at least $9/10$ in time $\tilde O(n^{2/3}) \cdot \qw_{O(n^*)}$ and with $O(n^{2/3})$ queries.  We can pair this with the noisy binary search algorithm of Feige et al.\ \cite{FeigeRPU94} to compute $R(\alpha, \Iset, k)$ with probability at least $9/10$ in time $\tilde O(n^{2/3}) \cdot \qw_{O(n^*)}$ and with $O(n^{2/3} \log n)$ queries.
\end{proof}

\begin{theorem}
\label{thm:comp}
Let $\Sigma$ be an alphabet and $\alpha \in \Sigma^{n^*}$.
For any constant integer $h \ge 2$, $\blds(\alpha, \Iset)$ is a 2-decomposable instance $(h, h)$-maximizing problem.  
When $s(\Iset) = n$, the create and completion functions can be computed by a quantum algorithm with oracle access to $\alpha$ in time $\tilde O(n^{2/3}) \cdot \qw_{O(n^*)}$ and with $O(n^{2/3} \log(n))$ queries.
\end{theorem}

\begin{proof}
Let $\indices = \{i_1, i_2, j_1, t\}$ with $s(\Iset)=n$.  We define the create function $\delta(\alpha, \indices, \ell)$ for $0 \le \ell \le h-1$.  The create function first computes $j_2$, the largest index in the interval $\{j_1,\ldots, t\}$ such that 
$\alpha[i_2] \concat \alpha[j_1 \upto j_2]$ is distinct.  This computation can be done by a quantum algorithm in time $\tilde O(n^{2/3}) \cdot \qw_{O(n^*)}$ and with $O(n^{2/3} \log(n))$ queries by \cref{fact:ed_binary}.

Let $n_2 = j_2 - j_1 + 1$.  Fix a partition of $\{j_1, \ldots, j_2\}$ into $h$ intervals $\Jset_0, \ldots, \Jset_{h-1}$, where each interval has size either $\ceil{n_2/h}$ or $\floor{n_2/h}$, and the endpoints of $\Jset_\ell$ can be computed from $j_1, j_2, \ell$ in time $O(\log n^*)$.  Let $x_\ell, y_\ell$ be respectively the left and right endpoint of $\Jset_\ell$.  As in \cref{cor:h}, further define 
\begin{align*}
p_\ell &= L(\alpha, i_1, i_2, j_1, j_2, \prev_{i_2, j_1}(x_\ell)), \\
q_\ell &= \min(i_2, p_\ell+|\Jset_\ell|-1), \\
y_\ell'&= \min(y_\ell, R(\alpha, i_1, i_2, j_1, j_2, \suc_{i_2, j_1}(q_\ell))) \enspace.
\end{align*}
With these definitions, we can define the create function as $\delta(\alpha, \Iset, \ell) = \{p_\ell, q_\ell, x_\ell, y_\ell'\}$.
Each of $p_\ell, q_\ell, y_\ell'$ can be computed in time $\tilde O(n^{2/3}) \cdot \qw_{O(n^*)}$ and with $O(n^{2/3} \log(n))$ queries by \cref{fact:ed_binary}, thus the create function can be computed in the same time.

We now define the completion function. The completion function first computes $n_2'$, the length of the longest distinct substring of $\alpha[j_1 \upto t]$ whose left endpoint is $j_1$.  This computation can again be done by a quantum algorithm in time $\widetilde O(n^{2/3}) \cdot \qw_{O(n^*)}$ and $O(n^{2/3} \log(n))$ queries by \cref{fact:ed_binary}.
Note that $\blds(\alpha, \Iset) = \max(n_2', \blds(\alpha, (i_1, i_2, j_1, j_2))$.
Let $v_\ell = \blds(\alpha, (p_\ell, q_\ell, x_\ell, y_\ell'))$ be the evaluation of $\blds$ on the instance defined by $\delta(\alpha, \Iset, \ell)$. Then 
\[
\gamma(\alpha, \Iset, \ell, \delta(\alpha, \Iset, \ell), v_\ell) = 
\max(n_2', i_2 - q_\ell + p_\ell - j_1 + v_\ell) \enspace.
\]
By the fact that $\blds(\alpha, \Iset) = \max(n_2', \blds(\alpha, (i_1, i_2, j_1, j_2))$ and \cref{cor:h} we then have that 
$\blds(\alpha, \Iset) = \max_{0 \le \ell \le h-1} \gamma(\alpha, \Iset, \ell, \delta(\alpha, \Iset, \ell), v_\ell)$ as required.  Thus $\gamma$ is a valid completion function and can be computed by a quantum algorithm in time $\widetilde O(n^{2/3}) \cdot \qw_{O(n^*)}$ and with $O(n^{2/3} \log(n))$ queries. 
\end{proof}

\begin{theorem}
\label{thm:qblds}
Let $\Sigma$ be a finite alphabet and $\alpha \in \Sigma^{n^*}$ be a string.  There is a quantum algorithm that for any $n^*$-valid 2-description $\Iset$ with $s(\Iset)=n$ computes $\blds(\alpha, \Iset)$ in time $\widetilde{O}(n^{2/3}) \cdot \qw_{O(n^*)} + O(\sqrt{n} \log(n^*))$ and with $O(n^{2/3} \cdot \log(n))$ queries.
\end{theorem}

\begin{proof}
For convenience, let $f(n) = \widetilde{O}(n^{2/3}) \cdot \qw_{O(n^*)}$ be an upper bound on the time complexity of the create and completion functions from \cref{thm:comp}.
By \cref{thm:rec} and \cref{thm:comp}, the time complexity $\overline{T}(n)$ for $\blds(\alpha, \Iset)$ with $s(\Iset)=n$ satisfies the following recurrence relation for any $h \ge 2$:
\begin{align*}
\overline{T}(1) &= O(\log n^* + \qw_{O(n^*)}) \\
\overline{T}(n) &\le c \sqrt{h}(\overline{T}(n/h) + f(n)) + f(n) \enspace. 
\end{align*}
By taking $h = 2c^6$, the solution to this recurrence becomes $\widetilde{O}(n^{2/3}) \cdot \qw_{O(n^*)} + O(\sqrt{n} \log n^*)$
as claimed.

For the query complexity, let $g(n) = O(n^{2/3} \log(n) \cdot \qw_{O(n^*)}$ be an upper bound on the query complexity of the create and completion functions from \cref{thm:comp}.
By \cref{thm:rec} and \cref{thm:comp}, the query complexity $T(n)$ for $\blds(\alpha, \Iset)$ satisfies the following recurrence relation for any $h \ge 2$:
\begin{align*}
T(1) &= O(1) \\
T(n) &= c \sqrt{h}(T(n/h) + g(n)) + g(n) \enspace.
\end{align*}
Again taking $h = 2c^6$, the solution to this recurrence becomes 
$O(n^{2/3} \cdot \log(n))$.
\end{proof}

\subsection{Application to \lds}
Finally, we can use the algorithm for \blds\ as part of a bottom-up algorithm for \lds.
\begin{theorem} \label{thm:lds}
There is a quantum algorithm that solves \lds\ in time $\widetilde O(n^{2/3}) \cdot \qw_{O(n)}$ and a quantum algorithm that solves \lds\ with $O(n^{2/3} \log(n) \log \log(n))$ queries.
\end{theorem}

\begin{proof}
We use the bottom-up approach.  Let the input be $a \in \Sigma^n$.  We can pad the input by repeating the last character without changing the length of a longest distinct substring.  Thus we may assume $n$ is a power of 2 without changing the asymptotic complexity.

Let $P(a, t, k)$ be the length of a longest distinct substring of $a$ contained in the interval $\{(k-1)n/2^t+1, \ldots , kn/2^t\}$, with the left endpoint in the left half of this interval and right endpoint in the right half of this interval.
Letting $i_1 = (k-1)n/2^t + 1, i_2 = (2k-1)n/2^{t+1}, j_1 = (2k-1)n/2^{t+1}+1, j_2 = kn/2^t$ and $\Iset = (i_1, i_2, j_1, j_2)$ we see that $P(a, t, k) = \blds(a, \Iset)$. 
Thus by \cref{thm:qblds} there is a quantum algorithm that outputs $P(a, t, k)$ with probability at least $9/10$ in time $\widetilde O(n^{2/3}/2^{2t/3} \cdot \qw_{O(n)})$ and with $O(n^{2/3}/2^{2t/3} \log(n/2^t)$ queries.   
Applying \cref{thm:bu_easy} gives a quantum algorithm for \lds\ with running time $\widetilde O(n^{2/3}) \cdot \qw_{O(n)}$ and a quantum algorithm with query complexity $O(n^{2/3} \log(n) \log \log(n))$.
\end{proof}

\section{Klee's Coverage}
\label{sec:klee}

Given a set $\cal B$ of $n$ axis-parallel hyperrectangles (boxes) in $d$-dimensional real space $\R^d$, the \KMf \ problem asks to compute the volume of the union of the boxes in $\cal B$. A special case is the \KCf \ problem where we are also given a base box $\Gamma$, and the question is to decide whether the union of the boxes in $\cal B$ covers $\Gamma$. 
In 2-dimensions the complexity of the \KMf \  problem is $O(n \log n)$~\cite{Klee77}, and for any constant $d \geq 3$,
Chan~\cite{Chan13} has given an $O(n^{d/2})$ time classical algorithm for it.

\begin{problem}[\KCf ]
Given a set $\cal B$ of $n$ axis-parallel boxes and a base box $\Gamma$ in $\R^d$, is $\Gamma \subseteq \cup_{B \in {\cal B}} B$?
\end{problem}

We give a quantum algorithm for $\KCf$ that achieves an almost quadratic speedup over the classical divide and conquer algorithm of Chan~\cite{Chan13}, when $d \geq 8$.
The speedup is less than quadratic for $5 \leq n \leq 7$, and there is no quantum speedup for $n \leq 4$.

\begin{theorem}
\label{thm:klee}
For every constant $\varepsilon > 0$, the quantum time complexity of the \KCf \ problem is $O(n^{d/4 + \varepsilon} \cdot \qw_{N(n)})$ when $d \geq 8$, and is $O(n^2 \cdot \qw_{N(n)})$ for $5 \leq d \leq 7$, where $N(n) = O(n^{d/2 + \varepsilon})$.
\end{theorem}

Informally it is easy to see why the problem can be solved by a divide and conquer algorithm: if the base box $\Gamma$ is divided into an arbitrary number of base sub-boxes, then $\Gamma$ is not covered by ${\cal B}$ if and only if there is a base sub-box which is not covered by ${\cal B}$. However, such a division should also ensure that for every base sub-box, the number of boxes in ${\cal B}$ intersecting it also gets appropriately reduced, a highly non-trivial task. Chan's classical divide and conquer algorithm achieves exactly that by using a weighting scheme, and our quantum divide and conquer algorithm essentially follows the classical algorithm of Chan. Still, we need a couple of modifications.
Besides the quantum procedures in our algorithm, we will divide an instance into $h$ sub-instances and not into $2$, where $h$ is a sufficiently large constant to be specified later. We also use padding in the create step to ensure the uniform size of the subproblems. 

\begin{proof}
We start with a few preliminary remarks. First observe that the boxes in $\cal B$ cover $\Gamma$ exactly when they cover $\Int{\Gamma}$, the interior of $\Gamma$.
By an instance of length $n$ we mean an instance $a = (\Gamma, {\cal B})$ where the number of boxes in $\cal B$ is $n-1$, with the base box $\Gamma$ bringing the total length of the instance to $n$.
As we have observed, $\KCf$ is a natural constructible instance $(h,m)$-disjunctive problem with a trivial completion function, where $h$ can be any constant, and given $h$, the value of the best $m$ requires some reasoning. 
Therefore the algorithm itself is fully specified by the create step that produces the $h$ constitutive strings. We suppose that the input size is $n^*$. 

In Chan's description, the create step contains two distinct parts, called {\em simplification} and {\em cut}, and we will complete this with a third part, {\em padding}.
In fact, the algorithm also contains a preprocessing step which should be accounted for in the final complexity.
In the preprocessing, we sort the input boxes in each dimension. Since, in the $\mathsf{QRAG}$ model, coordinates can be compared classically in time $O(\qw_{N(n^*)})$, this takes time $O(n^* \log n^* \cdot \qw_{N(n^*)})$.

In the simplification step, we transform the problem into an instance without {\em slabs}, that is without boxes in $\cal B$ which cover $\Gamma$ in each dimension but one.
Formally, a $k$-{\em slab} is a box whose restriction to $\Gamma$ is of the form $\{(x_1, \ldots , x_d) \in \Gamma : a_1 \leq x_k \leq a_2)\}$, for some real numbers $a_1 < a_2$. After the elimination of a $k$-slab, the $k$th coordinates of all the base boxes and the other boxes are changed accordingly, that is all regions between $a_1$ and $a_2$ are also eliminated.
The elimination of a slab might create new slabs, that is box which was not a slab before the elimination might become one after.
Therefore the elimination process is done iteratively. We start with the family of boxes containing all original slabs, we eliminate the elements of the family one by one, while after each elimination we add the newly created slabs. The process ends when we have an instance without any slab or we find a slab which fully covers $\Gamma$.
This can be done classically in $O(n^2 \cdot \qw_{N(n^*)})$ time for pre-sorted inputs.
In~\cite{Chan13} the simplification is described without dealing with the elimination of the freshly created slabs during the elimination process of the initial ones, and it is claimed that this can be done in time $O(n)$. While doing the overall elimination process in time $O(n^2 \cdot \qw_{N(n^*)})$ is easy, doing it in linear time requires further proof that we don't attempt here.
The slower elimination process changes the result only for $d<8$.

Without the elimination of the slabs, it is true that every box that doesn't fully cover $\Gamma$ has a $(d-1)$-face intersecting $\Int{\Gamma}$, but slabs don't have faces of dimension lower than this intersecting $\Int{\Gamma}$. 
After the elimination of the slabs, the remaining boxes have at least one $(d-2)$-face intersecting $\Int{\Gamma}$ which makes possible a better complexity analysis of the algorithm.

The cut step is done with the help of a weight function that we define for simplified instances.
Let ${\cal F(B)}$ denote the set of $(d-2)$-faces of the boxes in ${\cal B}$ that intersect $\Int{\Gamma}$.
For every $f \in {\cal F(B)}$ which is orthogonal to the $i$th and $j$th axes, we define its weight as $w(f) = h^{(i+j)/d}$, and we define the weight of the instance  $a = (\Gamma, {\cal B})$ as $w(a) =  \sum_{f \in {\cal F(B)}} w(f).$ 
The weight of a $(d-2)$ face is at most $h^2$, and the number of $(d-2)$-faces per box is a constant (depending on $d$). Therefore there exists a constant $\kappa >0$, also depending on $d$, such that for every instance $a$ of length $n$, we have $w(a) \leq \kappa n$.

The cut steps cycle over the $d$ dimensions.
From instance $a = (\Gamma, {\cal B})$ we create $h$ subproblems $b^{(k)} = (\Gamma_k , {\cal B}_k)$, for $k \in [h].$
We cut the base box $\Gamma$ into $h$ sub-boxes $\Gamma_1, \ldots  \Gamma_h$ by parallel hyperplanes which are orthogonal to some axis. 
For technical simplicity, all the cuts will be described as orthogonal to the first axis, and at the end of each iteration, the axes names $(1,2, \ldots, d)$ will be cyclically permuted as $(d, 1, \ldots , d-1)$.  Let $W$ be the sum of the weights of the $(d-2)$-faces in ${\cal F(B)}$ orthogonal to the first axis.
The cutting is done by $h-1$ hyperplanes $x_1 = v_1, \ldots , x_1 = v_{h-1}$ where the values $v_1, \ldots , v_{h-1}$ are chosen such that for every $1 \leq k \leq h$,  the total weight of the $(d-2)$ faces in ${\cal F(B)}$ orthogonal to the first axis and intersecting  $\Int{\Gamma_{k}}$ is at most $W/h$. 
The boxes belonging to $b^{(k)}$ are chosen as ${\cal B}_k = \{ B \in {\cal B} : B \cap \Int{\Gamma_{k}} \neq \emptyset \}.$ 
The cutting can be done classically in $O(n \cdot \qw_{N(n^*)})$ time.

After the cut step, the weight of a $(d-2)$-face parallel with the first axis expressed in the permuted axes' names
is multiplied by $h^{-2/d}$, while the weight of a $(d-2)$-face orthogonal to the first axis is multiplied by $h^{(d-2)/d}$.
Since the total weight of $(d-2)$ faces orthogonal to the first axis inside the interior of each $\Gamma_k$ is divided by a factor of $h$, altogether the total weight of the $(d-2)$-faces inside each sub-box decreases by a factor of $h^{2/d}$, that is $w(b^{(k)}) \leq w(a)/h^{2/d}$, for every $k \in [h]$.  

In order to be able to apply \cref{thm:recgen} we have to ensure that the constitutive strings have all the same length, and we enforce this via padding. Let us recall that $w(a) \leq \kappa |a|$, for some constant $\kappa > 0$, thus, for every $k \in [h]$, we have
$$
|b^{(k)}| \leq w(b^{(k)}) \leq w(a)/h^{2/d} \leq \kappa |a|/h^{2/d}.
$$
For every $k \in [h]$, we define the constitutive string $a^{(k)}$ by repeating some arbitrarily chosen box in $ {\cal B}_k$ an appropriate number times to make $|a^{(k)}|$ be equal to $\kappa |a|/h^{2/d}.$ Clearly negative instances remain negative and positive instances remain positive by this transformation.
We can now apply \cref{thm:recgen} with the constant $m = h^{2/d} \cdot \kappa^{-1}$ and we have an algorithm that solves the problem. 

Let $\overline{T}(n^*, n)$ denote the time complexity of the algorithm, maximized over all \emph{instances} of length $n$.
Then from \cref{thm:recgen} we have $\overline{T}(n^*, n) = O(\qw_{N(n^*)})$, when $n$ is small, and 
\[
\overline{T}(n^*, n) \leq c h^{1/2} (\log h + \overline{T}(n^*, n/ (h^{2/d} \cdot \kappa^{-1}))) + O(n^2 \cdot \qw_{N(n^*)})\enspace ,
\]
when $n$ is large.
For $d \geq 8$, the solution of this is
$$
\overline{T}(n^*, n) = O(n^{d/4 + \varepsilon} \cdot \qw_{N(n^*)}),
$$
if the constant $h$ is sufficiently large, for example if $h > \max\{c, \kappa\}^{d^2/\varepsilon}$.
For $5\le d \leq7$, we get $\overline{T}(n^*, n) = O(n^2 \cdot \qw_{N(n^*)})$. From a similar recursion for the memory size, it is easy to see that  $N(n) = O(n^{d/2+\epsilon})$ when $d\ge 5$.
The result follows when we take $n=n^*$.
\end{proof}

\section{Quantum complexity of some rectangle problems in $\APSPc$}
\label{sec:apsp}
\subsection{The problems and classical complexity}
\label{subsec:classical}

In this section, we will address the quantum complexity of two related maximization problems on $n^2$ weighted points arranged in the plane in a grid. 
They respectively look for a sub-rectangle with maximum sum, and for a sub-rectangle where the sum of the 4 corner values with alternating signs is maximum. We will need the following definitions where we think about the points as contained in a matrix.

Given an $n \times n$ matrix $B$ and four indices $1 \leq i \leq j \leq n$ and $1 \leq k \leq \ell \leq n$, we denote by $F_B((i,k),(j,\ell))$ the 4-{\em combination} $B_{ik} + B_{j \ell} - B_{i \ell} - B_{jk}$ and by $S_B((i,k),(j,\ell))$ the sum
$$\sum_{i \leq u \leq j, k \leq v \leq \ell} B_{uv}.$$

\begin{problem}[\MaxSubMf \ (\MaxSubM)]
\label{prob:maxsubm}
 Given an $n \times n$ matrix $B$ of integers find a maximum weight submatrix, that is $$\arg\max_{i \leq j, k \leq  \ell } S_B((i,k),(j,\ell)).$$
\end{problem}

\begin{problem}[\MaxFourCf \ (\MaxFourC)]
\label{prob:maxfourc}
 Given an $n \times n$ matrix $B$ of integers, find a maximum 4-combination, that is $$\arg\max_{i \leq j, k \leq  \ell } F_B((i,k),(j,\ell)).$$
\end{problem}

The $\MaxSubM$ (also called Maximum Subarray) problem is a special case of a basic problem in geometry, where given $N$ points in the plane, one has to find an axis-parallel rectangle that maximizes the total weight of the points it contains~\cite{DGM96, BCNP14}. The best-known algorithms run in time $O(N^2)$ for this general problem. The particular geometry of the $N =n^2$ points arranged in a grid allows faster solutions for the $\MaxSubM$ problem and the best algorithms for it run in time $O(N^{3/2}) = O(n^3)$~\cite{TT98, Tak02}. The $\MaxFourC$ problem was defined in~\cite{BDT16}, where it was shown that it can be solved in $O(n^3)$ time. It is also another generalization of the \SSST \ problem. 

Before stating our quantum complexity results, we will dwell on the classical complexity of these problems. 
In~\cite{BDT16} they are analyzed in the context of fine-grained complexity, and it is proven that in a strong sense, both $\MaxSubM$ and $\MaxFourC$ have the same deterministic complexity as the well-known $\APSPf$ problem.  
Fine-grained complexity establishes fine complexity classes among computational problems by picking some well-studied problem $P$ which can be solved in time $\widetilde{O}(T(n))$, for some polynomial $T(n)$, but for which no algorithm is known in time $T(n)^{1- \varepsilon}$, for any $\varepsilon > 0$.
Then the fine-grained (quantum) complexity class with respect to $P$ consists of all problems $X$ such that $P$ and $X$ are (quantum) {\em sub-$T(n)$ equivalent}, that is inter-reducible by some appropriate (quantum) {\em sub-$T(n)$ reductions}, ensuring that either, both, or neither of them can be solved in (quantum) time $T(n)^{1- \varepsilon}$, for some $\varepsilon > 0$. The existence of such reductions also implies that all problems in the class are solvable in time $\widetilde{O}(T(n))$. We won't give here the technical definition of the appropriate reduction.

The class $\APSPc$ of problems that are sub-$n^3$ equivalent in the above sense to $\APSP$ is one of the richest in fine-grained complexity theory~\cite{VW18, Vas19}. It contains various path, matrix, and triangle problems, and also $\MaxSubM$ and $\MaxFourC$.  We enumerate here some important problems from this class.

\begin{problem}[\APSPf]
\label{prob:apsp}
  Given a weighted graph on $n$ vertices, compute the length of a shortest path between $u$ and $v$, for all vertices $u,v$.
\end{problem}

\begin{problem}[\MP]
\label{prob:mp}
  Given two $n \times n$ matrices, compute their matrix product over the (min, +) semiring.
\end{problem}

\begin{problem}[\Met]
\label{prob:met}
  Given an $n \times n$ matrix, decide if it defines a metric.
\end{problem}

\begin{problem}[\MaxT]
\label{prob:maxt}
 Given a weighted graph on $n$ vertices, find a maximum weight triangle, where the weight of a triangle is the sum of its edges.
\end{problem}

\begin{theorem}[\cite{VW18, BDT16}]
{\bf Problems~\ref{prob:maxsubm}, \ref{prob:maxfourc}, \ref{prob:apsp}, \ref{prob:mp}, \ref{prob:met}, \ref{prob:maxt} }
are in $\APSPc$.
\end{theorem}

\subsection{Quantum complexity}
\label{subsec:quantum}

The study of fine-grained complexity equally makes sense in the quantum case, and indeed several recent works have studied quantum fine-grained complexity~\cite{ACLWZ20, BPS21, BLPS22}. These papers often address the quantum complexity of problems in the same classical equivalence class and~\cite{ABLPS22} has specifically considered $\APSPc$. Of course, it is not guaranteed at all that classically equivalent problems remain equivalent in the quantum model of computing, and this is exactly what happens
with $\APSPc$. While all known problems in the class receive some quantum speedup, the measure of the speedup can differ from problem to problem. It turns out that many of the problems in $\APSPc$ can be solved either in time $\widetilde{O}(n^{5/2})$ or in time $\widetilde{O}(n^{3/2})$ by simple quantum algorithms, and concretely $\APSP$ falls in the former category. In some cases, the classical fine-grained reductions also work for establishing the relevant quantum equivalences.
This is illustrated by the following propositions. 

\begin{proposition}
The problems $\MP$ and $\APSPf$  can be solved in quantum time $\widetilde{O}(n^{5/2}),$ and they are quantum { sub}-$n^{5/2}$ equivalent.
\end{proposition}

\begin{proof}
The upper bounds come from quantum minimum finding (\cref{fact:minimum}) and the classical sub-$n^3$ reductions of~\cite{FM71, Mun71} also establish the quantum sub-$n^{5/2}$ equivalence. 
These reductions show that if one of the problems can be solved in (classical or quantum) time ${O}(n^2 +T(n))$ then the other can be solved in time $\widetilde{O}(n^2 +T(n))$.
\end{proof}

\begin{proposition}
The problems $\MaxT$ and $\Met$ can be solved in quantum time $\widetilde{O}(n^{3/2})$, and they are quantum { sub}-$n^{3/2}$ equivalent.
\end{proposition}

\begin{proof}
The upper bound can be obtained by Grover search (\cref{fact:search}) over $n^3$ elements and the classical sub-$n^3$ reductions of~\cite{VW18} establish also the quantum sub-$n^{3/2}$ equivalence. Their reductions show that one of the problems can be solved in (classical or quantum) time ${O}(T(n))$ if and only if the other can be solved in time $\widetilde{O}(T(n))$.
\end{proof}

However, for some problems finding the ``right" quantum upper bounds or establishing fine-grained equivalences requires innovative quantum algorithms. For example~\cite{ABLPS22} used variable time quantum search and specific data structures to construct quantum algorithms for {\textsc{$\Delta$-Matching Triangles}} and {\textsc{Triangle Collection}}, two problems from $\APSPc$. 
In particular, they have shown that in the $\mathsf{QRAG}$ model the former problem can be solved in quantum time $\widetilde{O}(n^{3/2 + o(1)})$ when $\omega(1) \leq \Delta \leq n^{o(1)}$, and the latter problem has a quantum algorithm of complexity $\widetilde{O}(n^{3/2})$.

It turns out that, from the point of view of quantum complexity, the two rectangle problems of our concern are quite interesting. In their study, we will suppose that the absolute value of all input integers is polynomially bounded in the input size. As optimizing logarithmic factors are not relevant for this discussion, for simplicity, we choose to state our results taking both $\qr_{n^2}$ and $\qw_{O(n^2)}$ to be $O(\log^2n)$ when $N(n)$ is a polynomial in $n$ (\cref{sec:instances}).

We will show that the quantum time complexity of $\MaxSubM$ is $O(n^{2} \log^2 n)$ and that its quantum query complexity is $\Omega(n^2).$ One novelty of this is that, to our knowledge, this is the first example of a problem from $\APSPc$ whose quantum time complexity is not of the order of $n^{5/2}$ or $n^{3/2}$.
The second novelty is that our lower bound is established for the query complexity and therefore it is unconditional.
The possibility of proving a matching query lower bound is a specific feature of the quantum complexity since the classical time complexity of $\MaxSubM$ is believed to be super-linear in the input size.

The interesting aspect of the ${O}(n^{3/2} \log^{5/2} (n))$ quantum algorithm we obtain for $\MaxFourC$ is that it is based on quantum divide and conquer.
More precisely, we exploit the fact that the one-dimensional analog of this problem is  $\SSST$ for which, in \cref{thm:three}, we have designed an $O(n^{1/2} \log^{5/2}(n))$ time quantum algorithm based on that technique. 

\begin{theorem}
\label{thm:msm}
The quantum time complexity of $\MaxSubM$ is $O(n^{2}\log^2 n)$ in the $\mathsf{QRAG}$ model. 
Its quantum query complexity is $\Theta(n^2).$ 
\end{theorem}

\begin{proof}
Let $B$ be an $(n,M)$ matrix. We define the $(n+1,M+1)$ matrix $C$ by $C_{ik} = S_B((1,1),(i,k))$, for $1 \leq i,k \leq n$, and 
we set $C_{i0}=C_{0i} = 0$, for all $0 \leq i \leq n$.
Then we have $B_{ik} = C_{ik} - C_{i-1k} - C_{ik-1} + C_{i-1k-1}$.
Therefore $C_{ik} = B_{ik} + C_{i-1k} + C_{ik-1} - C_{i-1k-1}$ can be computed by dynamic programming in time $O(n^2\cdot\qw_{O(n^2)})$, as memory access and addition both take time $\qw_{O(n^2)}$ (see \cref{sec:instances}).
We have then, for $1 \leq i < j \leq n$ and $1 \leq k < \ell \leq n$, $$S_B((i,k),(j,\ell)) = F_C((i-1,k-1),(j,\ell)).$$
Therefore the maximum submatrix, over indices verifying $i < j $ and $k < \ell $, can be found in time  $O(n^2(\log  n+ \qw_{O(n^2)})) = O(n^2\log^2 n)$ by quantum maximum search (\cref{fact:minimum}). 

For the query lower bound we will make a reduction from the Boolean majority function $\Maj(x_1, \ldots, x_m)$ which is by definition $1$, if $\sum_{i=1}^m x_i > \floor{m/2}$.
It is known that the quantum query complexity of $\Maj$ on $m$ bits is $\Omega(m)$~\cite{BBCMW01}. Suppose we are given $n^2$ bits $x_{1}, \ldots x_{n^2}$ that we arrange into a matrix $B$.
We define a $(2n+1) \times (2n+1)$ matrix $C$ as follows. We put $B$ into the intersection of the first $n$ rows and first $n$ columns.
Into the intersection of the first $n$ rows and last $n$ columns we put an $n \times n$ matrix $D$ which has $\floor{n^2/2}$ 1's and 
$\ceil{n^2/2}$ 0's. We put $C_{j,n+1}= -n-1$ for $1 \leq j \leq n$, and we put -1 into the rest of $C$. Because of the negative barrier put into the column $n+1$ and the negative elements in the lower half of $C$, its maximum submatrix is either fully inside $B$ or fully inside $D$.  
Therefore $\Maj(x_1, \ldots, x_{n^2}) = 1$ if and only if the maximum submatrix in $C$ has value greater than $\floor{n^2/2}$.
\end{proof}

\begin{theorem}
\label{thm:mfc}
The quantum time complexity of \ \MaxFourC ~is ${O}(n^{3/2} \log^{5/2} n)$ in the $\mathsf{QRAM}$ model. 
\end{theorem}

\begin{proof}
Let us be given an $n \times n$ matrix $B$. 
For every $1 \leq i  \leq j \leq n$, we define the array $A^{ij}$ by $A^{ij} = B_{j} - B_{i}$, where $B_{i}$ denotes the $i$th row of $B$.
Then
$$
B_{ik} + B_{j \ell} - B_{i \ell} - B_{jk} = A^{ij}_{\ell} - A^{ij}_k,
$$
and therefore 
$$
\max_{i \leq j, k \leq  \ell } F_B((i,k),(j,\ell)) = \max_{i \leq j} \max_{k \leq \ell } A^{ij}_{\ell} - A^{ij}_k.
$$
By \cref{thm:three}, we can determine $\max_{k < \ell } A^{ij}_{\ell} - A^{ij}_k$, in time $O(\sqrt{n\log n}\cdot\lambda_2(n,\qr_{n^2}),$ for  fixed $i \leq j$.   
Let $m$ be this value, then $\max_{k \leq \ell } A^{ij}_{\ell} - A^{ij}_k = \max\{m,0\}$, taking into account also the case
$k=\ell$.
Therefore by \cref{cor:meo} we can run quantum maximum finding on the indices $i \leq j$ for the function $f(i,j) = \max_{k \leq \ell } A^{ij}_{\ell} - A^{ij}_k$ in time O($n^{3/2}(\log n)^{1/2}\cdot \lambda_2(n,\qr_{n^2})) = O(n^{3/2}\log^{5/2}n)$.
\end{proof}

We observe that the classical fine-grained reduction of~\cite{BDT16} from $\MaxT$ to $\MaxFourC$ is also a sub-$n^{3/2}$ reduction.
Indeed, they show that if $\MaxFourC$  can be solved in (classical or quantum) time $\widetilde{O}(T(n))$ then $\MaxT$ can be solved in time $\widetilde{O}(n +T(n))$.
Therefore if we believe that there is no sub-$n^{3/2}$ quantum algorithm for $\MaxT$ then this also applies for $\MaxFourC$. However, proving an $\Omega (n^{3/2})$ query lower bound for $\MaxT$ could be a hard if not impossible task. Indeed, while the best-known quantum algorithm for finding a triangle in an unweighted graph has time complexity $O(n^{3/2} \log n)$, query algorithms of much lower complexity are known for this problem~\cite{MSS07, Bel12, LMS17, Gall14}, and the situation could be similar for $\MaxT$. 
We leave it as an open problem whether proving such a query lower bound for $\MaxFourC$ is possible.

\begin{openq}
What is the quantum query complexity of $\MaxFourC$?
\end{openq}

\subsection*{Acknowledgments}
This research is supported by the National Research Foundation, Singapore, and A*STAR under its CQT Bridging Grant and its Quantum Engineering Programme under grant NRF2021-QEP2-02-P05. AB is supported by the Latvian Quantum Initiative under European Union Recovery and Resilience Facility project no. 2.3.1.1.i.0/1/22/I/CFLA/001 and the QuantERA project QOPT.
TL is supported in part by the ARC DP grant DP200100950.  TL thanks CQT for their hospitality during a visit where part of this work was conducted.
\bibliographystyle{alpha}
\bibliography{references}

\end{document}